\documentclass[a4paper]{paper}
\usepackage{graphicx}
\usepackage[latin1]{inputenc}
\pdfoutput=1
\usepackage{amsmath,amssymb,amsthm}

\newcommand{\E}{\mathcal{E}}

\renewcommand{\E}{\mathbb{E}}

\DeclareMathOperator{\Var}{Var}

\newtheorem{prop}{Proposition}[section]

\newtheorem{theo}[prop]{Theorem}

\makeatletter
 \renewcommand*{\@font@warning}[1]{}
 \makeatother
             
\def\DD{\displaystyle} 
\DeclareMathOperator*{\lra}{\longrightarrow}

\DeclareMathOperator*{\argmin}{argmin}

\newcommand\egaldef{\stackrel{\mbox{\upshape\tiny def}}{=}}
\newcommand{\1}{\leavevmode\hbox{\rm \small1\kern-0.35em\normalsize1}}
\newcommand{\ind}[1]{\1_{\{#1\}}}
\def\DD{\displaystyle} 
\date{Mai 2010}

\author{
Cyril Furtlehner
  \thanks{INRIA Saclay -- LRI, Bat. 490, Universit\'e Paris-Sud -- 91405
  Orsay cedex (France)}
\and 
Jean-Marc Lasgouttes
  \thanks{INRIA Paris Rocquencourt -- Domaine de Voluceau B.P.\ 105 -- 78153 Le Chesnay cedex (France)}
\and 
Maxim Samsonov\footnotemark[1]
}

\title{One-dimensional Particle Processes with Acceleration/Braking Asymmetry}
\begin{document}
\maketitle 

\abstract{The slow-to-start mechanism is known to play an important role 
in the particular shape of the Fundamental diagram of traffic and 
to be associated to hysteresis effects of traffic flow.
We study this question in the context of exclusion and queueing processes,
by including an asymmetry between deceleration and 
acceleration in the formulation of these processes. For exclusions processes, this 
corresponds to a multi-class process with transition asymmetry 
between different speed levels, while for queueing processes we consider non-reversible 
stochastic dependency of the service rate w.r.t the number of clients.
The relationship between these 2 families of models is analyzed on the ring geometry, 
along with their steady state properties. Spatial condensation phenomena
and metastability is observed, depending on the level of the aforementioned asymmetry.
In addition we provide a large deviation formulation of  the fundamental diagram (FD) 
which includes the level of fluctuations, in the canonical ensemble
when the stationary state is expressed as a product form of such generalized queues.}

\section{Introduction}

In the study of models for traffic~\cite{Schrek}, properties of the
FD, which gives  a relation between the traffic flux and the density or
alternatively the dependence between the speed and the flux, or the
speed and the density, play an important role. 
In the three phases traffic theory of Kerner~\cite{Kerner}, the
traffic phase diagram on highways consists of the free flow, the
synchronized flow and the wide moving jam. In the free flow regime, at
low density, the flux is simply proportional to the density of cars;
in the congested one, at large density, massive clusters of cars are
present, and the flux decreases more or less linearly with this
density; in the intermediate regime, the relation between flux and
density is largely of stochastic nature, due to the presence of a
large amount of small clusters of cars propagating at various random
speeds. It is not clear however whether in this picture these phases,
and especially the synchronized flow phase, are genuine dynamical or
thermodynamical phases, meaningful in some large size limit in the
stationary regime, or are intricate transient features of a slowly
relaxing system. In fact there is still controversy about the reality
of the synchronized phase of Kerner at the moment~\cite{ScHe}.

In order to analyze properties of the FD from a statistical physics perspective, 
we  look for a simple extension of the totally asymmetric exclusion process (TASEP) 
as well as  the zero range process able to encode the fact that vehicles in the traffic 
may accelerate or brake. Some asymmetry between these two actions
is empirically known to be responsible of the way spontaneous congestion occurs. 
This is referred as the slow-to-start mechanism, not present in the original 
cellular automaton of Nagel-Schreckenberg~\cite{NaSch}, but in its refined version 
like the velocity dependant randomized one (VDR) \cite{VDR} which exhibits a first 
order phase transition between the fluid and the congestion phase and hysteresis 
phenomena \cite{Blank} associated to metastable states.
So a realistic stochastic model should have in particular the property that
some spontaneous symmetry breaking among identical vehicles can occur,
as can be seen experimentally on a ring geometry for example~\cite{Sugiyama}.
With such a model at hand we would like to provide a method to compute the
FD and study emergence of non-trivial collective behaviors at macroscopic level.

The paper is organized as follows: in Section~\ref{sec:masep} 
we start by defining a simple and somewhat minimal
generalization of TASEP  and describe some of  its basic properties. 
In Section~\ref{sec:stoch-queues} a family of zero-range processes is introduced
which service rate depends stochastically on the state of the queue and on which 
there are relevant mappings of multi-speed TASEP;
we determine sufficient conditions for having a product form for the invariant 
measure of such processes coupled in tandem. 
In Section~\ref{sec:fd} we provide a  large deviation formulation giving 
the FD along with fluctuations on the ring geometry in the canonical ensemble, 
when the steady state has a product form.
Finally, in Section~\ref{sec:single-jam} 
we analyze the jam structure of a generalized
queueing process corresponding to the model introduced in Section~\ref{sec:masep}.
and interpret direct simulations of the multi-speed TASEP in this light.

\section{Multi-speed exclusion processes}\label{sec:masep}
\subsection{Model definition}\label{sec:model}
The model we investigate is a multi-type exclusion process,
generalizing the simple exclusion process on the line, introduced by
Spitzer in the 70's~\cite{Ligett,Spitz}, combined with some feature of
the Nagel-Schreckenberg cellular automaton~\cite{NaSch}.

In the totally asymmetric version of the exclusion 
process~\cite{Ligett,Spitz} (TASEP),
particles move randomly on a 1-d lattice, always in the same
direction, hopping from one site to the next one within a time interval
following a Poisson distribution and conditionally that the next site
is vacant. In the Nagel-Schreckenberg cellular automaton, the dynamics
is in parallel: all particles update their positions at fixed time
intervals, but their speeds are encoded in the number of steps that
they can take.  This speed can adapt stochastically, depending on the
available space in front of the particle.

The model that we propose combines the braking and accelerating
feature of the Nagel-Schreckenberg models, with the locality of the
simple ASEP model, in which only two consecutive sites do interact at
a given time. The trick is to allow each car to change stochastically
its hopping rate, depending on the state of the next site.  For a
$2$-speed model, let $A$ denotes sites occupied by fast vehicles, $B$
sites occupied by slow ones and let $O$ denote empty sites; the
model is defined by the following set of reactions,
involving pairs of neighbouring sites:

\begin{align}
AO &\lra^{\mu_{a}} OA \qquad \text{simple move of fast vehicle}\label{AOOA} \\[0.1cm]  
BO &\lra^{\mu_{b}} OB \qquad \text{simple move of slow vehicle}\label{BOOB} \\[0.1cm]
BO &\lra^{\gamma} AO  \qquad \text{slow vehicle spontaneously accelerates}\label{BOAO}     \\[0.1cm]
A\bar O &\lra^{\delta} B\bar O\qquad A\ \text{brakes behind}\ \bar O=A\ \text{or}\ B  \label{ABBB}
\end{align}
$\mu_a$, $\mu_b$,  $\gamma$ and $\delta$ denote the transition
rates, each transition corresponding to a Poisson event. The dynamics
is therefore purely sequential as opposed to the dynamics of the
Nagel-Schreckenberg model. It encode the fact that a slow vehicle tends to accelerate
when there is space ahead, while in the opposite case, corresponding
to (\ref{ABBB}) it tends to slow down.
The main mechanism behind congestion, namely the asymmetry between
braking and acceleration is potentially present in the model
when $\gamma$ is different from $\delta$.
Our model  is in fact similar to the model of Appert-Santen~\cite{ApSa}, in which 
there is a single speed, but particles have 2 states: at rest and moving
with possible transitions between these 2 states.

To define fully the model, its boundary conditions have to be
specified, either periodic in the ring geometry or with edges. In the
latter case, additional incoming and outgoing rates have to be
specified, depending on whether we want to model a traffic light a
stop or simply a segment of highway for example. In this paper we will
consider only the ring geometry. In the following we will refer to 
this specific model as the acceleration-braking totally asymmetric exclusion
process (ABTASEP).
 
\subsection{Known special cases}
Let us first notice that this model contains and generalizes several
sub-models which are known to be integrable with particular rates. The
hopping part (\ref{AOOA},\ref{BOOB}) of the models is just the
TASEP when
$\mu_{a}=\mu_{b}$, which is known to be integrable. Its
generalization to include multiparticle dynamics with overtaking is
the so-called Karimipour model~\cite{Cantini,Karim}
which turns out to be integrable as well. 
Matrix Ansatz method allows in some cases, like
the ASEP with open boundary conditions, to describe the stationary
regime of the models, using the representations of the so-called
diffusion algebras~\cite{DeEvHaPa}.  The acceleration/deceleration
dynamics is equivalent to the coagulation/decoagulation models, which
are known to be solvable by the empty interval method and by free
fermions for particular sets of rates \cite{SaFuLa}, but 
the whole process is presumably not integrable.

\subsection{Relation to tandem queues}\label{sec:qmap}
In some cases, the model can be exactly reformulated in terms of
generalized queueing processes (or zero range processes in the
statistical physics parlance), where service rates of each queue
follows as well a stochastic dynamics~\cite{tgf07}. In this previous
work we however considered exclusion processes involving three
consecutive sites interactions in order to maintain the homogeneity of 
labels in particle clusters. The mapping works only on the ring geometry, 
by identifying queues either
with
\begin{itemize}
\item[(i)] cars: clients are the empty sites.
\item[(ii)] empty sites: clients are the vehicles;
\end{itemize}
In our case the mapping of type (i) is exact. In the corresponding queueing process, 
queues are associated either with  fast or slow cars, having then service rates
$\mu_a$ or $\mu_b$. Slow queues become fast at rate $\gamma$ conditionally 
to having at least one client, while empty fast queues become slow at rate
$\delta$.

The mapping of type (ii), is more informative with respect to jam 
distribution, but  is not possible with transitions limited to
$2$-consecutive sites interactions, because in that case
homogeneity is not maintained in the clusters, and additional information 
to the number of cars and the rate of the car leaving the queue is 
needed to know the service rate of the queue. For this mapping  we have
to resort to some approximate procedure as shall be exemplified in 
Section~\ref{sec:single-jam}.

Another way to circumvent this would be to start from a 
a slightly different definition of our initial model, by not attaching
speed labels to cars, but instead to empty sites. Let $A$ and $B$ denote empty site with respectively 
fast and slow speed, and $V$ denote sites occupied with a vehicle, the set of transitions reads then
\begin{equation}\label{def:model2}
\begin{cases}
\DD VA \lra^{\mu_{a}} AV \qquad \text{simple move on fast site} \\[0.1cm]  
\DD VB \lra^{\mu_{b}} BV \qquad \text{simple move on slow site} \\[0.1cm]
\DD VA \lra^{\delta} VB  \qquad \text{fast site become slow}     \\[0.1cm]
\DD \bar V B \lra^{\gamma} \bar VA\qquad \text{slow site become fast}          
\end{cases}
\end{equation}
where $\bar V$ is an unoccupied site i.e. either $A$ or $B$.
The mechanism in this model is that the  empty sites visited more recently have a
slower associated speed  than others, leading possibly to congestion instability.
The mapping to tandem queues is then straightforward: queues are associated to empty 
site of type $A$ or $B$ with corresponding service rate $\mu_a$  or $\mu_b$; fast queues
with at least one client have a probability per unit of time $\delta$ to become a slow queue, 
while slow queue which are empty become fast  with probability rate $\gamma$.
This model generalizes directly to any speed levels. Interestingly, one see that 
in such a model  the service rate of queues are somehow related to the lifetime of clusters.
In this paper we will focus on  model of Section~\ref{sec:model}, deferring for future work 
the study of model~(\ref{def:model2}).

\subsection{Numerical observations}
Based on numerical simulations on the ring geometry, we make some
observations concerning the phenomenology of the
model, depending on the parameters. This is illustrated on
Figure~\ref{fig:st-plot}.
\begin{figure}[ht]
\centering
\includegraphics[height=3cm,width=12cm,angle=0]{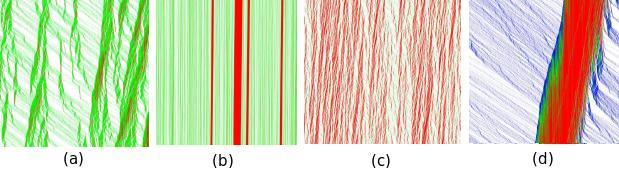}\\
\caption{\label{fig:st-plot} Space-time diagram 
  with 2 speed levels (a)-(c) and 3 in (d). Time is   going downward and particles to the right.  
Red, green and blue
  represent different speeds in increasing order.  The size of the
  system is $3000$ except for (b) where it is $100000$. Setting are
  $\mu_a=100$, $\mu_b = 10$, $\gamma = 10$, $\delta=1$
  for (a) and (b) and $\delta=10$ for (c), all with density $\rho
  = 0.2$.}
\end{figure}
When no asymmetry between braking and accelerating is present, as in
TASEP on a ring, no spontaneous large jam structure is observed.  
As the density $\rho$ of cars increases, one observes a smooth transition
between a TASEP of fast particles for small $\rho$ to a TASEP of slow
particles around $\rho\simeq 1$. Instead, when the ratio
$\delta/\gamma$ is reduced, there is a proliferation of small jams.
Below some threshold of this ratio, we observe a condensation
phenomenon, associated to some critical value of the density: above
this critical value, after a time which presumably scales as
some power of $L$, one or more large jams, which absorb a finite fraction of
the total number of cars may develop. If there are many of such large
jams, a competition occurs, which end up possibly in the long term with one
single large condensate depending on the relative  value of $\gamma$ and $\mu_b$ 
and also of the size $N$ of the system. It is tempting to interpret this 
as a condensation mechanism at equilibrium in the canonical ensemble~\cite{EvMaZi}, 
by combining the Nagel-Paczuski~\cite{NaPa} interpretation
of competing queues with some results in a previous work~\cite{tgf07},
which in the context of tandem queues on a ring, allows this
condensation mechanism to take place if the apparition of slow vehicle is a
sufficiently rare event. As such, the mechanism for congestion is 
then understood qualitatively as follows:
assume that a jam of size $s$ appears, created somewhere by a mutation
$A\lra B$. Let $t(s)$ be the expected waiting time in this queue. The
probability that a vehicle is still of type $A$ when it reaches the
end of the queue is simply
\[
p(\mu =\mu_a) = e^{-\delta t(s)} = p_s.
\]
Assuming that $s$ does not change much during $t(s)$, then we have
\[
t(s) = s(\frac{1}{\mu_a}p_s +\frac{\mu_a+\gamma}{\mu_a(\mu_b+\gamma)}(1-p_s)),
\]
if we take into account the fact that the vehicle may accelerate at rate $\gamma$ before leaving the 
queue. Anyway, this gives $t(s)$ self-consistently and yields qualitatively that, at
the beginning of the process, there is a distribution of jams with
various effective service rates. As time evolves, either long-lived
jams with effective service rates $\mu_a(\mu_b+\gamma)/(\mu_a+\gamma)$ 
are able to survive, then new small jams can never appear 
because their effective service rates are strictly larger, so in the end there is a competition
between the existing jams; the largest one eventually remains alone after
waiting a large amount of time. Typically this happens when $\delta \ll \mu_b$ and
$\delta < \gamma$ (Figure~\ref{fig:st-plot}(b)).
If the situation with a single jam is not stable, then
no large jam may develop at all, and only small fluctuations are to be
observed. When  $\delta > \gamma$ one likely observes the kind of jams seen in
Figure~\ref{fig:st-plot}(c).

\begin{figure}[ht]
\centering
\includegraphics[scale=0.2]{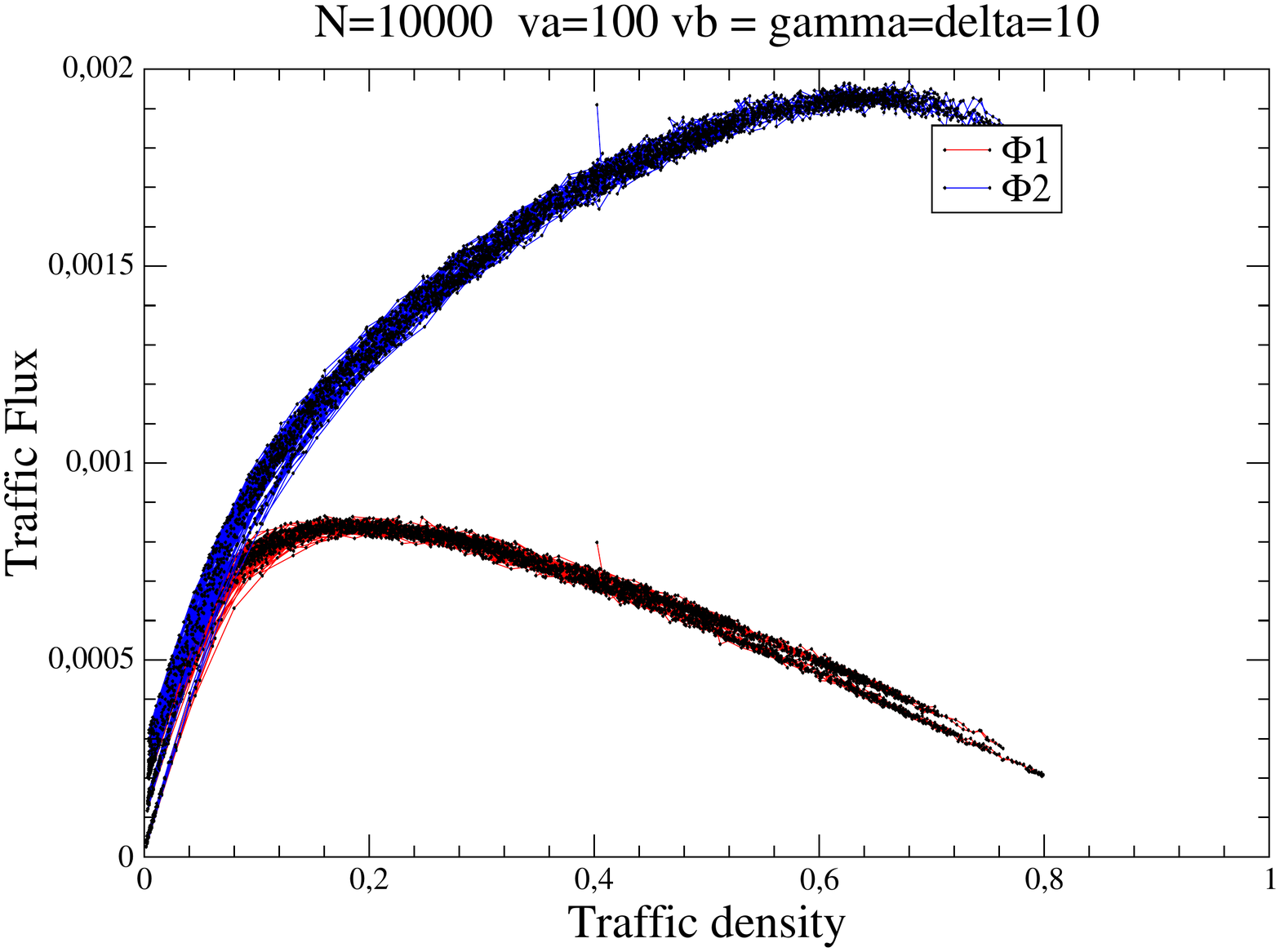}
\includegraphics[scale=0.2]{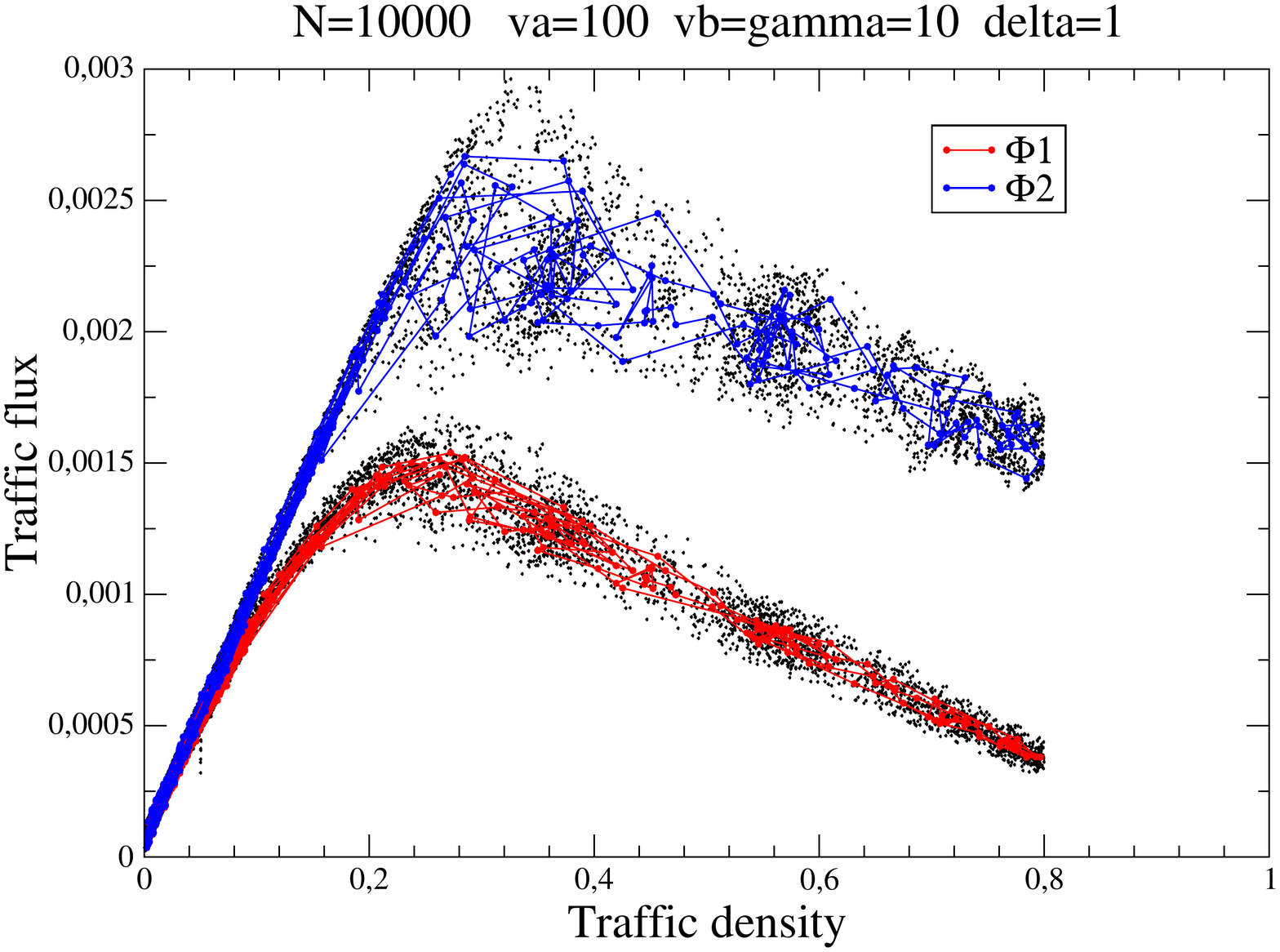}
\caption{\label{fig:fdiag0} Fundamental diagram with $\gamma=\delta$ (left) and
$\gamma=10\delta$ (right) with small spontaneous rate of emission and escape of particles.
Trajectories of $\Phi_1\egaldef \sum_{i=0}^{N-1} \mu_a A_iO_{i+1}+\mu_b B_iO_{i+1}$ and
$\Phi_2\egaldef \sum_{i=0}^{N-1} \mu_a A_i+\mu_b B_i$ are displayed.}
\end{figure}

The effect of the asymmetry is also clearly observed on Figure~\ref{fig:fdiag0},
where we allow particle to enter or quit the system with some very low probability rate compared 
to others. The global density of cars then performs a random walk, and by looking at trajectories in the 
FD plane, we see hysteresis effect for $\delta < \gamma$. 
We have also simulated a model with $3$ speed levels, see
Figure~\ref{fig:st-plot}(d). In that case, small jams may have
different speeds, depending on which type of slow car is leading. Then
a cascade mechanism takes place, slow speed regions generate even
slower speed clusters of cars and so on, and some kind of synchronized
flow is observed.

To conclude this section, let us finally remark that
Figure~\ref{fig:st-plot}(a) is very reminiscent of
coagulation-decoagulation process, which is somewhat expected from the
previous discussion, since it is present in the equations. One could
therefore expect the jam structure to share some properties with the
directed percolation process taking place on the ring geometry.

\subsection{Hydrodynamic equations and solitons}
Although it is not clear whether the hydrodynamic limits is valid in our context, 
one can at least write down the corresponding equations on the density $\rho_a(x,t)$
and $\rho_b(x,t)$, with $x=i/N$ and $N$ large,  expected in such a limit:
\begin{align}
\partial_t\rho_a + v_a\partial_x(\rho_a\bar\rho) &= -\delta\rho_a\rho + 
\gamma\rho_b\bar\rho\label{eq:hydroa}\\[0.2cm]
\partial_t\rho_b + v_b\partial_x(\rho_b\bar\rho) &= \delta\rho_a\rho - \gamma\rho_b\bar\rho\label{eq:hydrob}
\end{align}
with $\rho=\rho_a+\rho_b$ the total density of cars, $\bar\rho = 1-\rho$
the density of empty sites and
\[
v_{a,b} \egaldef \lim_{N\to\infty}\frac{\mu_{a,b}^{(N)}}{N}
\]
if the rates $\mu_{a,b}^{(N)}$ are allowed to be rescaled when $N$ is varied.
This is an hyperbolic system of equations since the 
matrix  
\[
M = \left[
\begin{matrix}
v_a(1-2\rho_a-\rho_b) & -v_a\rho_a\\
-v_b\rho_b & v_b(1-\rho_a-2\rho_b)
\end{matrix}\right]
\]
has two real eigenvalues 
\[
\lambda^{\pm} = \frac{1}{2}\bigl(v_a+v_b-v_a(2\rho_a+\rho_b)-v_b(\rho_a+2\rho_b)\pm\sqrt\Delta,
\]
with 
\[
\Delta = \bigl(v_a(1-2\rho_a-\rho_b)-v_b(1-\rho_a-2\rho_b)\bigr)^2+4v_a v_b\rho_a\rho_b\ >0.
\]
The method of characteristic applied to this system is however not giving much insight into it, 
because of the non-linear shape of the characteristics. 
\begin{figure}[ht]
\centerline{\resizebox*{.5\textwidth}{!}{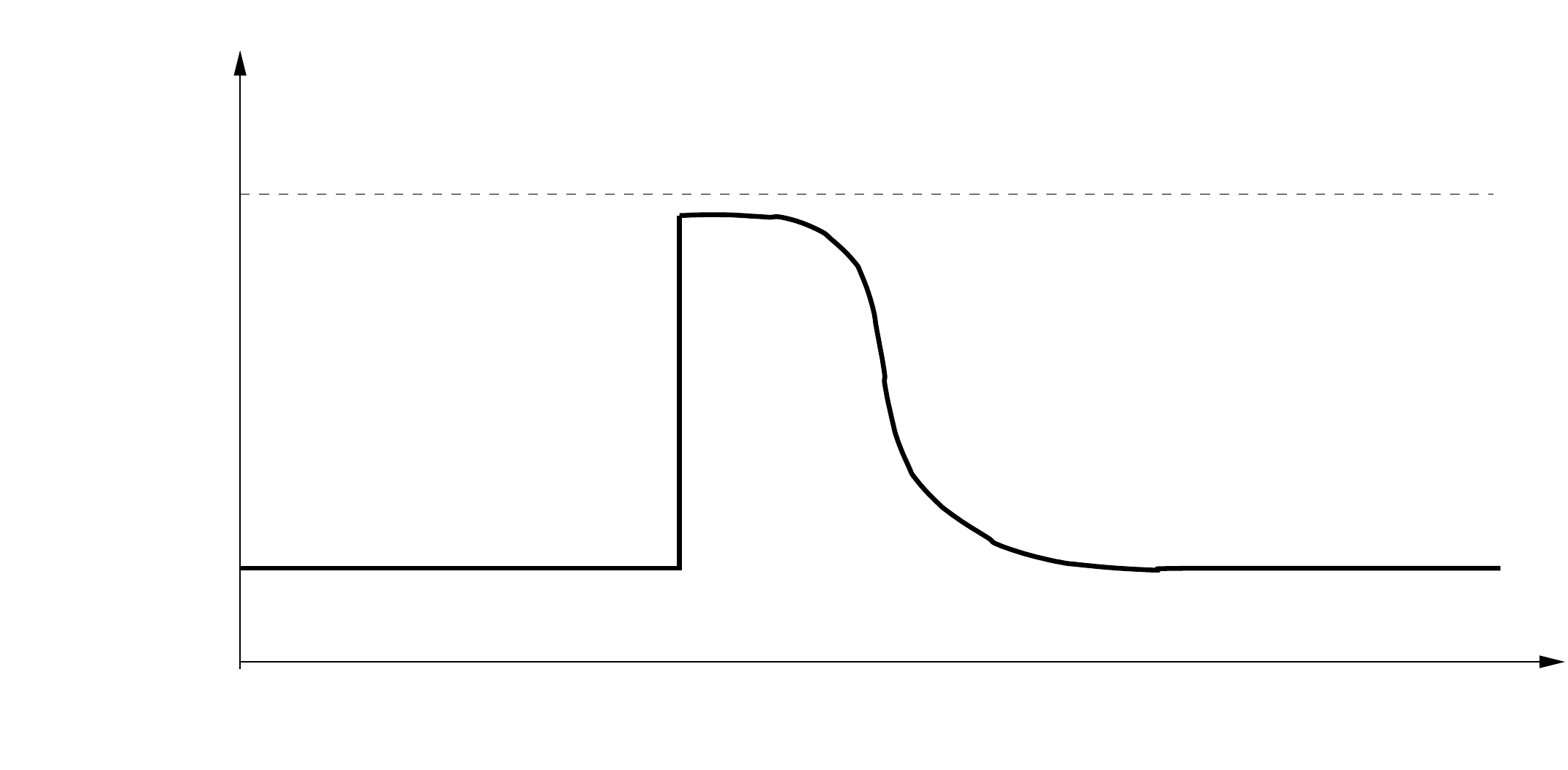}}
\caption{\label{fig:soliton}}
\end{figure}
At steady state we can nevertheless look for backward travelling waves solutions, of the 
form $\rho_{a,b}(x+vt)$, with $v$ the propagation velocity. 
Replacing $\partial_t$ by $v\partial_x$ in (\ref{eq:hydroa},\ref{eq:hydroa})
and summing up the two equations one can check that the flux
\begin{equation}\label{eq:phi}
\Phi \egaldef (v_a\rho_a + v_b\rho_b)\bar\rho + v\rho,  
\end{equation}
is a conserved quantity
\[
\frac{d\Phi}{dx} = 0.
\]
This allows to express $\rho_a$ and $\rho_b$ as function of $\rho$:
\begin{equation}
\rho_a = \frac{\Phi-\rho v-v_b\rho\bar\rho}{(v_a-v_b)\bar\rho}\qquad\text{and}\qquad
\rho_b = \frac{\Phi-\rho v-v_a\rho\bar\rho}{(v_b-v_a)\bar\rho}\label{eq:rhob}
\end{equation}
After rescaling w.r.t $v_a$ the model has only $3$ independent parameters:
\[
\eta\egaldef \frac{v_b}{v_a},\qquad \kappa\egaldef \frac{\delta}{\gamma},\qquad
\nu\egaldef\frac{\gamma}{v_b},
\]
and we rescale as well the soliton speed $\upsilon\egaldef v/v_a$ and the flux
$\phi\egaldef\Phi/v_a$.
Replacing $\rho_a$ and $\rho_b$ with their expressions in (\ref{eq:rhob}), 
we obtain an equation relating $\partial_x\rho$ and $\rho$:
\[
\frac{d\rho}{dx} = -C \bar\rho\ \frac{P(\rho)}{Q(\rho)},
\]
with 
\[
C = \frac{\nu}{2}(1-\eta\kappa)
\]
and where $P$ and $Q$ are both polynomials of degree $3$:
\begin{align*}
P(\rho)&\egaldef \bar\rho(\rho-\rho^-)(\rho-\rho^+)+
\bigl(\frac{\nu}{2}(1-\kappa)\bar\rho+\kappa\nu\bigr)(\phi-\upsilon)\\[0.2cm]
Q(\rho)&\egaldef \bar\rho^2(\rho-\rho_1)+(\phi-\upsilon)\frac{\upsilon}{2\eta}
\end{align*}
where
\[
\rho_1 = 1+\frac{1+\eta}{\eta}\upsilon,\qquad
\rho^\pm =\frac{1+(1-\kappa)\upsilon\pm\sqrt{\Delta}}
{2(1-\eta\kappa)},
\]
are the roots when $\phi=\upsilon$
with the discriminant
\[
\Delta = (1-(1+\kappa)\upsilon)^2+4(\eta-\upsilon)\kappa\upsilon,
\]
assumed to be positive. Assuming $\kappa<1$ we have $\rho^-\in[0,1]$ and $\rho^+\ge 1$
[resp. $\rho^+\in]\rho^-,1]$] when $\upsilon\le\eta$ [resp. $\upsilon\ge\eta$].
$1/C$ gives the length scale of the interface between 
regions of different densities. In particular if $\mu_{a,b}$ are maintained at a fixed 
value the interface becomes a step function when $N$ is large.

Let $\rho_{max}$ and $\rho_{min}$ the extremal values of the density. We expect a solution
which decreases from $\rho_{max}$ to $\rho_{min}$ with $R(\rho)$ staying strictly positive  
for $\rho\in[\rho_{min},\rho_{max}]$. 
Since $R(\rho)$ is a single valued function of $\rho$, there must be a discontinuity, i.e. a shock to 
relate these two values as depicted on Figure~\ref{fig:soliton}. We don't expect everything
to be discontinuous at this point, in particular, as long as the shock is not 
at rest ($v=0$) the proportion of fast cars $u\egaldef \rho_a/\rho$
should be continuous, since cars are slowing down at finite pace $\delta$.
This means that $\rho_{min}$ and $\rho_{max}$ are not 
independent. Instead, letting $u^*$ the value of $u$ at the shock, from (\ref{eq:rhob}), 
they are both solution of ($\bar u^* \egaldef 1-u^*$)
\[
(v_au^*+v_b\bar u^*)\rho^2-(v_au^*+v_b\bar u^*+\upsilon)\rho+\phi = 0.
\]
Given that the sum and product of the two roots have to be smaller than $2$
and $1$ respectively, we deduce that $\upsilon$ and $\phi$ are both strictly smaller 
than $v_au^*+v_b\bar u^*$. In addition the constraint $\rho_{max}\le 1$ implies
$\phi\ge \upsilon$.
As a consequence, for $\rho\in[0,1[$ we have $Q(\rho)>0$,  vanishing only at 
$\rho=1$ in the limit case $\Phi=\upsilon$. Concerning $P(\rho)$ 
it has then $3$ roots $\rho_i, i\in\{1,2,3\}$ with 
$0<\rho_1<\rho_2<1<\rho_3$. 
To obtain a general solution, there are two free parameters, $\upsilon$ and $\rho_{min}$ and therefore 
$2$ constraints,  associated to fixing the length  of the system to one 
the mean density to $\rho_0$. These read
\[
\int_{\rho_{min}}^{\rho_{max}} \frac{d\rho}{\bar\rho}\frac{Q(\rho)}{P(\rho)} = C\qquad\text{and}\qquad
\int_{\rho_{min}}^{\rho_{max}} d\rho\frac{\rho}{\bar\rho}\frac{Q(\rho)}{P(\rho)} = C\rho_0.
\]
When condensation occurs we expect to observe a full congestion ($\rho=1$) at some point.
This entails from (\ref{eq:phi}) that $\phi=\upsilon$, 
yielding in that case the equation:
\[
\frac{d\rho}{dx} = C\frac{(\rho^+-\rho)(\rho-\rho^-)}{\rho-\rho_1}.
\]
and 
\[
\rho_a = \frac{\upsilon-\eta\rho}{1-\eta}\qquad\text{for}\qquad \rho<1.
\]
On one hand, we should have $\rho^+>1$ to be able to have a point where $\rho=1$, 
consequently $\upsilon$ is necessarily smaller than $\eta$;
on the other hand $\rho_a$ becomes negative when $\rho\to 1$ if $\upsilon<\eta$. 
As a result, there is no consistent way to obtain genuine 
condensation in this hydrodynamic framework.

\section{Tandem queues with stochastic service rates}\label{sec:stoch-queues}
\subsection{Definition of the process}
As already mentioned in Section~\ref{sec:qmap}, in some cases, on the ring 
geometry,  exclusion processes can be mapped onto a queueing process. These have a fixed number 
of queues, organized in tandem and 
in each queue, both the number of clients and the service rate are stochastic, the service 
rate being not necessarily a function of the number clients. In some case, each queue taken 
in isolation may not be a reversible process,  with some hysteresis feature relevant to traffic.
This family of models is defined more formally as follows:
let ${\cal G} = ({\cal N}, {\cal L})$ a network of queues with dynamical
(stochastic) service rates.  By dynamical service rates, we 
mean that each single queue $i\in {\cal N}$ is represented by a vector
$z_i(t) = (n_i(t),\mu_i(t))\in E_i \subset {\mathbb N}\times
{\mathbb R}^+$.  $n_i(t)$ is the number of clients and $\mu_i(t)$ is a
service rate, which represents the global transition rate from $z_i$
to $z'_i = (n_i-1,\mu_i') \in V_i^-(z_i)$ ($V_i^-(z)$ is the set of
points in $E_i$ having one client less than $z$).
This process is a generalization of queueing processes with 
stochastic service~\cite{Harris,tgf07,SaFuLa}.
From the viewpoint of the statistical physics it is an extension of the
zero-range process, obtained by adding internal dynamics. 

Two sets of transition probability matrices $p_i^\pm(z,z')$ and one
set of transition rates $q_i^0(z,z')$ are introduced to complete the
definition of the process.  When a client get served in queue $i$, there
is a transition from the state of the departure queue $z_i$ to one of the state 
$z'\in V_i^-(z)$ with a rate $\mu_i(t)p_i^-(z,z')$; in the destination queue, a transition 
occurs from the state $z_{i+1}$ into one of state $z'\in V_{i+1}^+(z)$
with probability $p_{i+1}^+(z,z')$. We have the normalizations,
\begin{equation}\label{eq:norma}
\sum_{z'\in V_i^\pm(z)}p_i^\pm(z,z') = 1.
\end{equation}
Additional internal transitions are allowed, where the service rate $\mu_i$
of queue $i$ changes
independently of any arrival or departure. The intensities of these transitions
are given by the set $q_i^0(z,z'),\ z'\in V_i^0(z)$ 
of transition rates
($V_i^0(z)$ is the set of points in $E_i$ having the same number of clients
as $z$).
The combined set of transition rates,
\[
q_i(z,z') \egaldef \lambda p_i^+(z,z')\ind{z'\in V_i^+(z)} + 
\mu(z) p_i^-(z,z')\ind{z'\in V_i^-(z)} +q_i^0(z,z')\ind{z'\in V_i^0(z)},
\]
defines for each $i\in {\cal N}$ a continuous time
Markov process representing the dynamics of one queue taken  in isolation with arrival rate $\lambda$.
The joint process  is then entirely specified by the state graph of the
single queue which set of nodes $\{(n,\mu)\}$ is a subset of ${\mathbb N}\times{\mathbb R}_+$, and 
in which each transition is represented by an oriented edge (see e.g. Figures~\ref{fig:balance}
and \ref{fig:composite}).

\subsection{Product form of clusters at steady-state}\label{sec:prodform}
In the stationary regime, the invariant measure of a network of reversible queues
is known to factorize into a product form in general~\cite{Kel}, which allows to 
compute many equilibrium quantities explicitly.
With a dynamical rate it is interesting to check for such properties and 
in this section  we study the conditions under which the stationary
state has a product form. For simplicity we restrict the analysis 
to the case where $\cal G$ corresponds to a tandem of $N$ queues where the clients 
leave queue $i$ to enter in queue $(i+1)\mod(N)$, for any $i=0,\ldots,N-1$.  
We can establish the following sufficient conditions:
\begin{theo} 
Let $\pi_i^{\lambda}$ denotes the steady state probability
corresponding to queue $i$ taken in isolation and fed with a Poisson
process with rate $\lambda$. If the following partial balance
equations are satisfied (see Figure~\ref{fig:balance}),
\begin{align}
\sum_{z\in V_i^+(z_i)}  \mu(z)p_i^-(z,z_i) \pi_i^{\lambda}(z)\  &=
\lambda \pi_i^{\lambda}(z_i),\label{eq:partbalance1}\\[0.2cm]
\mu(z_i)\pi_i^{\lambda}(z_i) 
\ + \sum_{z\in V_i^0(z_i)} q_i^0(z_i,z)\pi_i^{\lambda}(z_i)
&=\nonumber\\[0.2cm] 
\sum_{z\in V_i^-(z_i)}\lambda p_i^+(z,z_i) \pi_i^{\lambda}(z)
&+ \sum_{z\in V_i^0(z_i)}q_i^0(z,z_i)  \pi_i^{\lambda}(z),\label{eq:partbalance2}
\end{align}
then the joint probability measure of the network has the 
following product form at steady state:
\begin{equation}\label{eq:prodform}
P(S = \{z_i, i\in {\cal N}\}) = \frac{\prod_{i\in{\cal N}} \pi_i^{\lambda}(z_i)}{P(\sum_i n_i = N)}
\end{equation}
\end{theo}
\begin{proof}
Using (\ref{eq:norma}), the global balance equations reads
\begin{align*}
\sum_{i\in{\cal N}}\Biggl[ \sum_{\substack{z'\in V_i^+(z_i)\\ z''\in V_{i+1}^-(z_{i+1})}}
\mu(z')p_i^-(z',z_i)&p_{i+1}^+(z'',z_{i+1})
\frac{\pi_i^\lambda(z')\pi_{i+1}^\lambda(z'')}{\pi_i^\lambda(z_i)\pi_{i+1}^\lambda(z_{i+1})}\\[0.2cm]
&+\sum_{z\in V_i^0(z_i)} q_i^0(z,z_i)\frac{\pi_i^\lambda(z)}{\pi_i^\lambda(z_i)}\Biggr]P(S= \{z_i\}) \\[0.2cm]
&=\ \sum_{i\in{\cal N}} 
\Bigl[\mu(z_i) + \sum_{z\in V_i^0(z_i)} q_i^0(z_i,z)\Bigr] P(S= \{z_i\})
\end{align*}
To find a sufficient condition select term $i$
in the first term of the left hand side and terms $i+1$ otherwise in
the preceding equation, yielding $\forall i\in {\cal N}$,
\begin{align*}
\sum_{\substack{z'\in V_i^+(z_i)\\z''\in V_{i+1}^-(z_{i+1})}}
\mu(z') p_i^-(z',z_i)p_{i+1}^+(z'',z_{i+1}) &\pi_i^{\lambda}(z') \pi_{i+1}^{\lambda}(z'') 
+\\
&\sum_{z\in V_{i+1}^0(z_{i+1})}
q_{i+1}^0(z,z_{i+1})\pi_i^{\lambda}(z_i) \pi_{i+1}^{\lambda}(z)\\[0.2cm]
=
\mu(z_{i+1}) \pi_i^{\lambda}(z_i) \pi_{i+1}^{\lambda}(z_{i+1})
&+\sum_{z\in V_{i+1}^0(z_{i+1})}
q_{i+1}^0(z_{i+1},z)\pi_i^{\lambda}(z_i) \pi_{i+1}^{\lambda}(z_{i+1}).
\end{align*}
After multiplying this last equation by $\lambda$, this sufficient 
condition rewrites
\begin{align*}
\Bigl(\sum_{z'\in V_i^+(z_i)}
\mu(z') p_i^-(z',z_i)\pi_i^{\lambda}(z')\Bigr)
&\Bigl(\lambda\sum_{z''\in V_{i+1}^-(z_{i+1})}
p_{i+1}^+(z'',z_{i+1})  \pi_{i+1}^{\lambda}(z'')\Bigr) \\[0.2cm]
&= \lambda\pi_i^{\lambda}(z_i)\times\\
\Bigl(
\mu(z_{i+1}) \pi_{i+1}^{\lambda}(z_{i+1})
+\sum_{z\in V_{i+1}^0(z_{i+1})}
&q_{i+1}^0(z_{i+1},z) \pi_{i+1}^{\lambda}(z_{i+1})
-q_{i+1}^0(z,z_{i+1})\pi_{i+1}^{\lambda}(z)\Bigr),
\end{align*}
which is satisfied $\forall i\in {\cal N}$ under the two sets of
sufficient conditions (\ref{eq:partbalance1}) and
(\ref{eq:partbalance2}).
\end{proof}

Note that reversible processes are special cases of processes obeying
(\ref{eq:partbalance1},\ref{eq:partbalance2}) and, in this respect,
our result is an adaptation to our context of the general results of
Kelly concerning product forms in queueing
networks~\cite{Kel}.  The next obvious question is then whether there
exists any such process which is non-reversible.

\subsection{Examples}
\begin{figure}
\begin{center}
\resizebox*{0.8\textwidth}{!}{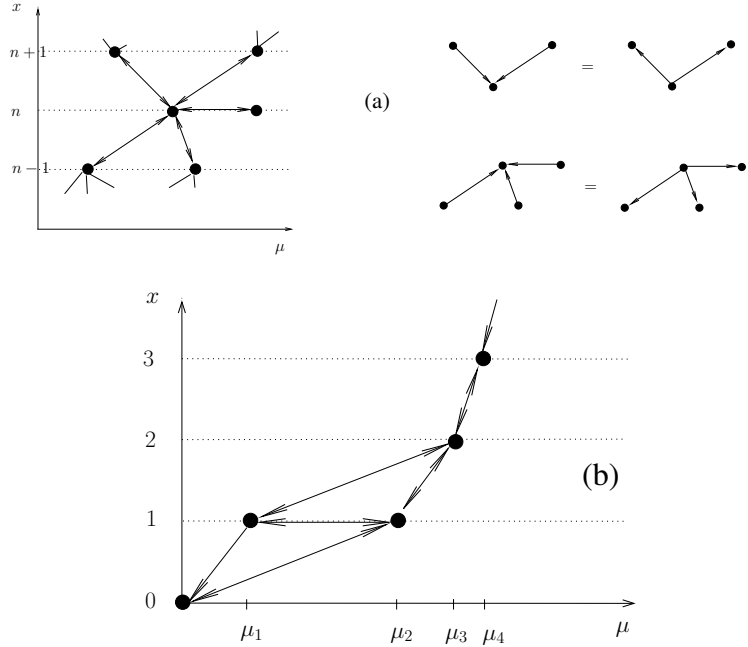}
\end{center}
\caption{(a) Partial balance condition corresponding to
  (\ref{eq:partbalance1})
  and (\ref{eq:partbalance2}). (b) Example
  of a non-reversible single queue process obeying  to this condition.}\label{fig:balance}
\end{figure}
Let us give two examples of non-reversible queues yielding a product form.
The isolated-queue state-graph corresponding to the first one is sketched in Figure~\ref{fig:balance}(b).
Using the labelling $(0)= (0,\mu_0)$, $(1)= (1,\mu_1)$, $(2)= (1,\mu_2)$, $(3)= (2,\mu_4)$ and
$(i)=(n-2,\mu_i)$ for $i\ge4$ for the different states of the isolated queue,
these partial  balance conditions written for the nodes $(i), i<4$ yield four 
independent relations out of six, between the single 
queues steady states probabilities $\pi_i$:
\begin{align}
\lambda \pi_0 = \mu_1\pi_1+ \mu_2\pi_2, &\qquad \lambda \pi_1 = \mu_3p_{31}\pi_3,\nonumber\\[0.2cm]
q_{21}\pi_2 = (q_{12}+\mu_1)\pi_1,&\qquad \lambda \pi_2 = \mu_3p_{32}\pi_3,\nonumber
\end{align}
which are actually compatible iff:
\[
\frac{p_{32}}{p_{31}} = \frac{q_{12}+\mu_1}{q_{21}}.
\]
Letting $p\egaldef p_{31} =1-p_{32}$, this corresponds to 
\[
p = \frac{q_{21}}{\mu_1+q_{12}+q_{21}}.
\]
For $i\ge 4$, detail balance holds anyway.
When this hold the joint measure of the tandem queue takes the product form~(\ref{eq:prodform})
whereas this example is by construction  never reversible
(transitions $0\to 1$ is absent), with the single queue invariant measure  
simply reading
\begin{align*}
\pi_1  &= \frac{\lambda q_{12}}{\mu_1\mu_2 + q_{12}\mu_2 + q_{21}\mu_1}\pi_0,\qquad
\pi_2  = \frac{\lambda(\mu_1+q_{21})}{\mu_1\mu_2 + q_{12}\mu_2 + q_{21}\mu_1}\pi_0,\\[0.2cm]
\pi_3  &= \frac{\lambda^2(\mu_1+ q_{12}+q_{21})}{\mu_3(\mu_1\mu_2 + q_{12}\mu_2 + q_{21}\mu_1)}\pi_0,\qquad
\pi_{n\ge 4}  = \biggl[\prod_{i=4}^n\frac{\lambda}{\mu_i}\biggr]\frac{\lambda
q_{12}}{\mu_1\mu_2 + q_{12}\mu_2 + q_{21}\mu_1}\pi_0,
\end{align*}
where $\pi_0$ is set upon normalization. 
\begin{figure}
\begin{center}
\resizebox*{0.45\textwidth}{!}{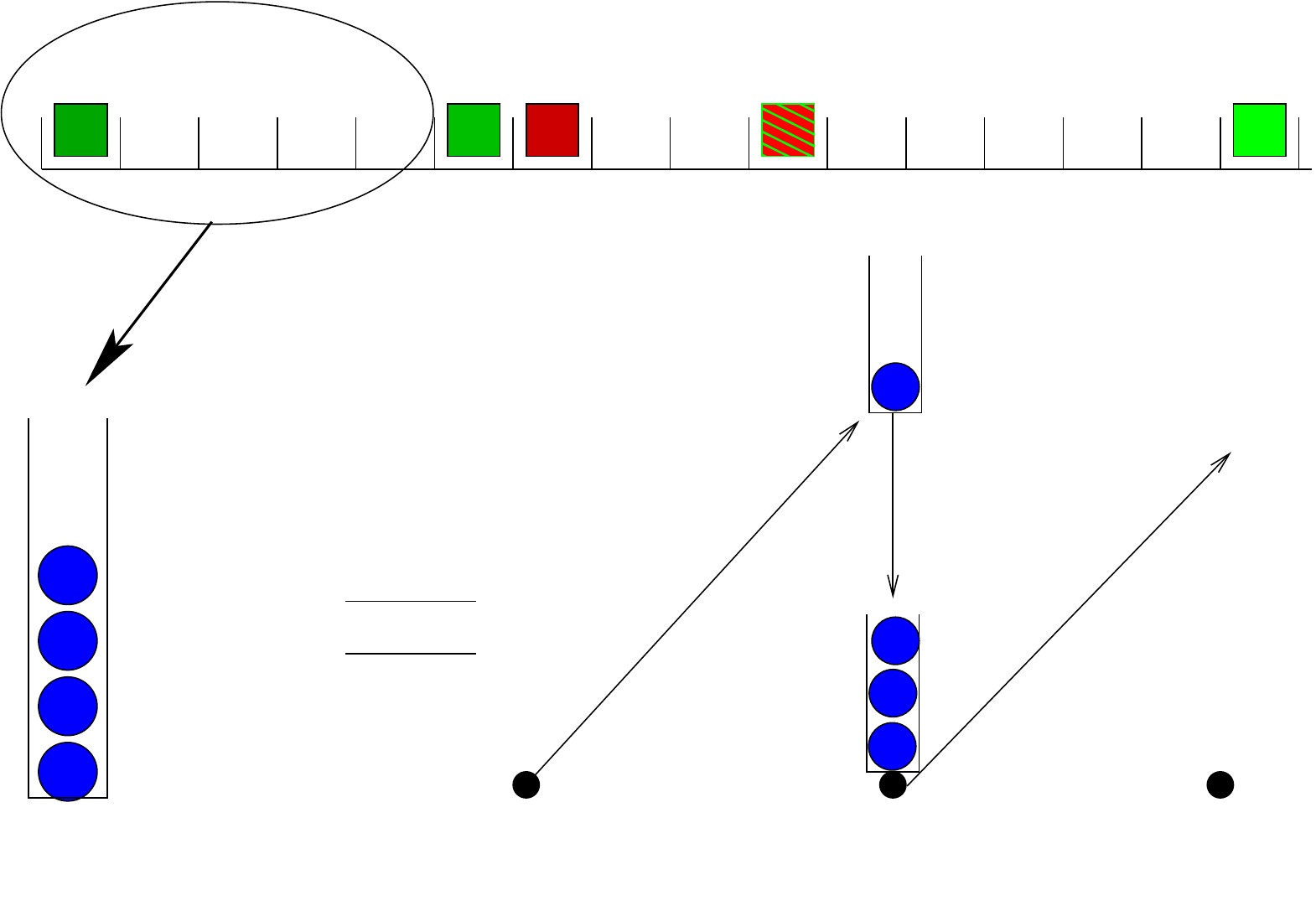}
\hspace{0.5cm}\resizebox*{0.4\textwidth}{!}{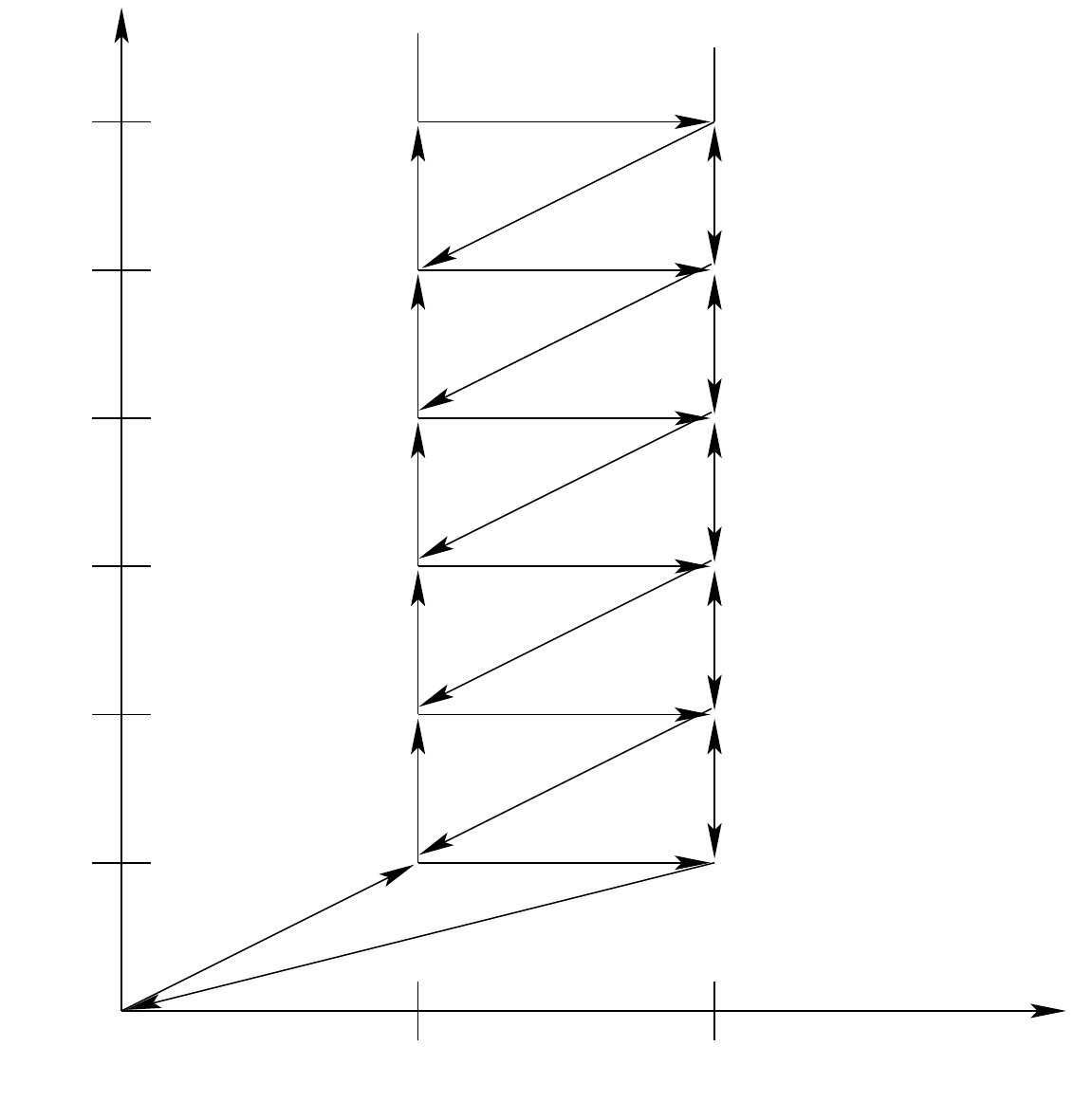}
\end{center}
\caption{\label{fig:composite}Composite traffic tandem queue process. 
Equivalence with exclusion process (left). Corresponding process in isolation,
$p=\rho_1^{n-1}\frac{\rho_1-\rho_2}{\rho_1^n-\rho_2^n}$ (right)}
\end{figure}

Our second example is depicted in Figure~\ref{fig:composite}. It
represents an exclusion process on a ring, where particles, to hop to the next site 
have to follow a cascading process. For sake of simplicity let us limit ourselves 
to a two stage cascade, where empty site are of two different types:
either blocking type $B$ and non-blocking type $A$.
The transition rules then read:
\begin{align}
VA &\lra^{\mu_1} BV \qquad \text{vehicle hopping } \\[0.1cm]  
VB &\lra^{\mu_2} VA \qquad \text{blocking site becomes non-blocking} \\[0.1cm]
AB &\lra^{\mu_2} AA  \qquad \text{blocking site becomes non-blocking}
\end{align}
This model is mapped on a composite queueing process as follows:
cars represent composite queues, with client arriving in the type $B$ 
(first stage of the composite queue)
and able to jump to the second stage in type $A$ when they reach the server of the queue, 
i.e. when the 
left neighbor is of type $A$. 
The whole tandem of queues has a  product form since each single queue taken in isolation 
is reversible.
Considering then the  queue depicted on Figure~\ref{fig:composite}(b) as an effective equivalent to 
the composite one, it is clearly non-reversible. Nevertheless
the partial balance sufficient condition for having a product
form are fulfilled. Indeed in the example of the figure the isolated
effective queue has for invariant measure
\begin{align*}
\pi(n,\tilde\mu = \frac{\mu_1\mu_2}{\mu_1+\mu_2}) &= (1-\rho_1)(1-\rho_2)\rho_1^n,\\[0.2cm]  
\pi(n,\tilde\mu = \mu_1) &= (1-\rho_1)(1-\rho_2)\frac{\rho_1^n-\rho_2^n}{\rho_1-\rho_2},  
\end{align*}
with $\rho_{1,2} = \lambda/\mu_{1,2}$. Computing the fluxes yields
then the desired partial balance conditions. It seems that this procedure of replacing 
composite of reversible  queues with effective one yield automatically non-reversible 
queues fulfilling 
the partial balance conditions.

The analysis of the fundamental diagram follows then almost straightforwardly when such
product form is present, although some non-trivial feature are also
expected depending on the value of the spectral gap between the
steady-state and the first excited state.

\section{Fundamental diagram for product form distributions}\label{sec:fd}
\subsection{Large deviation formulation}
In practice, points plotted in experimental \textsc{fd} studies are
obtained by averaging data from static loop detectors over a few
minutes (see e.g.~\cite{Kerner}).  This is difficult to compute from our
queue-based model, for which a space average is much easier to obtain.
The equivalence between time and space averaging is not an obvious
assumption~\cite{Blank}, but since jams are moving, space and time correlations are
combined in some way~\cite{NaPa} and we consider this assumption to be
quite safe. In what follows, we will therefore compute the \textsc{fd}
by considering the conditional probability
measure $P(\phi|d)$ for a closed system, where $d$ represents the spatial density of
cars and $\phi$ the normalized flux:
\[
\begin{cases}
\DD d = \frac{N}{N+L},\\[0.2cm]
\DD \phi = \frac{\Phi}{N+L},
\end{cases}
\ \text{with}\ 
\begin{cases}
L& \text{number of queues}\\
\DD N = \sum_{i=1}^L n_i&\text{number of vehicles}\\[0.2cm]
\DD \Phi = \sum_{i=1}^L \mu_i\ind{n_i>0}&\text{integrated flow}
\end{cases}
\]
The numbers $N$ of vehicles and $L$ of queues are fixed, meaning 
in the statistical physics parlance, 
that we are working with the canonical
ensemble. If we assume that we are in the conditions of having a product form 
presented in the preceding section,
for the stationary distribution with individual probabilities $\pi^\lambda(n,\mu)$ associated
to each queue taken in isolation then, taking into account  
the  constraints yields the following form of the
joint probability measure:
\[
P(\{n_i,\mu_i\}) = \frac{\delta\bigl(\textstyle N-\sum_{i=1}^L n_i\bigr)}{Z_L[N]} 
\prod_{i=1}^L \pi^\lambda (n_i,\mu_i),
\]
with the canonical partition function
\[
Z_L[N] \egaldef {\sum_{\{n_i,\mu_i\}} 
\delta\bigl(\textstyle N-\sum_{i=1}^L n_i\bigr)}\prod_{i=1}^L \pi^\lambda (n_i,\mu_i),
\]
where $\delta$ denotes now the usual Dirac function. When $\phi$ is interpreted as 
a continuous variable, the properly 
normalized density-flow conditional probability distribution takes the form
\begin{equation}\label{eq:fd}
P(\phi|d) =  \frac{L}{1-d}\frac{Z_L\bigl[L\frac{d}{1-d},L\frac{\phi}{1-d}\bigr]}
{Z_L\bigl[L\frac{d}{1-d}\bigr]},
\end{equation}
with
\begin{equation}\label{def:ZLNPhi}
Z_L[N,\Phi] \egaldef {\sum_{\{n_i,\mu_i\}}  
\textstyle \delta\bigl(N-\sum_{i=1}^L n_i\bigr)
\delta\bigl(\Phi-\sum_{i=1}^L \mu_i\ind{n_i>0}\bigr)}\prod_{i=1}^L \pi^\lambda (n_i,\mu_i),
\end{equation}
Note (by simple inspection, see e.g. \cite{Kel})
that $P(\phi|d)$ is independent of $\lambda$. 
$Z_L[N]$ and $Z_L[N,\Phi]$ represent respectively 
the probability of having $N$ vehicles and the joint probability for having at the 
same time $N$ vehicles and a flux $\Phi$, under the unconstrained product form. 
Under this product form, on general ground, we expect $d$ and 
$\phi$ to satisfy a large deviation principle (see~\cite{Touch} 
for a recent review connecting large deviations to statistical physics), 
i.e. that there exist two rate functions $I(d)$ and $J(d,\phi)$ such that for large $L$
\begin{align*} 
Z_L(N) \asymp e^{-LI(d)},\\[0.2cm]
Z_L[N,\Phi] \asymp e^{-LJ(d,\phi)},
\end{align*}
such that we expect a the large deviation version of the fundamental diagram of the form
\begin{equation}\label{eq:ldfd}
P(\phi|d) \asymp e^{-LK(\phi\vert d)}.
\end{equation}
with 
\[
K(\phi\vert d)\egaldef J(d,\phi)-I(d).
\]
When there is one single constraint like for $Z_L(N)$, the large deviation 
expression can be obtained by saddle point techniques like in \cite{EvMaZi,FaLa}.
For more than one constraint it seems easier to work variationally.  
Consider $L_{n,\mu}$ the number of queues having $n$ clients and service rate $\mu$.
Assuming $\mu$ runs over discrete values, the partition function (\ref{def:ZLNPhi}) can be recast in:
\begin{align}
Z_L(N,\Phi) &= \sum_{\{L_{n,\mu}\}}\frac{L!}{\prod_{n,\mu}L_{n,\mu}!}\prod_{n,\mu} 
\pi^\lambda(n,\mu)^{L_{n,\mu}}\nonumber\\[0.2cm]
&\times
\delta\bigl(L-\sum_{n,\mu}L_{n,\mu}\bigr)
\delta\bigl(N-\sum_{n,\mu}n L_{n,\mu}\bigr)
\delta\bigl(\Phi-\sum_{n>0,\mu}\mu L_{n,\mu}\bigr)\label{eq:znl}
\end{align}

As can be seen explicitly in the following, 
when there is no condensation, i.e. no heavy tail in the $L(n,\mu)$ distribution, 
we are in the condition of the  Sanov theorem~\cite{Touch} to get an asymptotic expression of 
the partition function.
Namely, 
\begin{equation}\label{eq:ldzln}
Z_L(N,\Phi) \asymp \sum_{\{y(n,\mu)\}}\exp\Bigl(-L{\cal F}[y]\Bigr)
\delta\bigl(C_1[y]\bigr)\delta\bigl(C_2[y]\bigr)
\delta\bigl(C_3[y]\bigr),
\end{equation}
with 
\[
y(n,\mu) \egaldef \frac{L_{n,\mu}}{L}
\]
the fraction of queues having $n$ clients and service rate $\mu$,
${\cal F}$ the large deviation functional and $C_{1,2,3}[y]$ the constraints:
\begin{align*}
C_1[y] &= \sum_{n,\mu}y(n,\mu)-1,\\[0.2cm]
C_2[y] &= \sum_{n,\mu} n y(n,\mu)-\frac{d}{1-d},\\[0.2cm]
C_3[y] &= \sum_{n>0,\mu} \mu y(n,\mu)-\frac{\phi}{1-d}.
\end{align*}
This is obtained explicitly by 
making use of the Stirling formula in~(\ref{eq:znl}), keeping only the leading terms to get:
\begin{equation}\label{eq:freen}
{\cal F}(y) = \sum_{n,\mu}y(n,\mu)\log\frac{y(n,\mu)}{\pi^\lambda(n,\mu)}.
\end{equation}
The rate functions are then obtained using the contraction principle~\cite{Touch}: 
\begin{align*}
I(d) &= \inf_{\{y:C_1[y]=0,C_2[y]=0\}} {\cal F}[y]\\[0.2cm]
J(d,\phi) &= \inf_{\{y:C_1[y]=0,C_2[y]=0,C_3[y]=0\}} {\cal F}[y]
\end{align*}
which apply here simply because $\cal F$ is convex and the constraints are linear.
The variational principle consists in to look for a distribution $y$
minimizing the mutual information with the product form while satisfying the constraints;
in this sense it is equivalent to the standard mean field approximation. 
A simple generalization would then 
amount to consider, in case where the product form is not valid, the 2-server problem to 
build a joint-law based on pairs $\pi^\lambda(n_1,\mu_1;n_2,\mu_2)$ 
of nearest neighbor dependencies and to use the Bethe approximation.

Let us introduce the moment and cumulant generating function $g$ and $h$  
associated to $\pi^\lambda$, 
\[
g(s,t) \egaldef \sum_{n=0,\mu}^\infty \pi^\lambda(n,\mu)e^{sn+t\mu} 
\qquad\text{and}\qquad h(s,t) \egaldef \log[g(s,t)],
\]
where it is assumed by convention that the rate $\mu$ is zero in absence of client.
The stationary point then reads
\begin{equation}\label{eq:varsol}
y(n,\mu) = \frac{\pi^\lambda(n,\mu)e^{n\lambda_n + \mu\lambda_\mu }}{g(\lambda_n,\lambda_\mu)},
\end{equation}
with the Lagrange multipliers  $\lambda_n$ and $\lambda_\mu$ associated to constraints $C_2$
and $C_3$, are implicitly given by 
\begin{equation}\label{eq:constraints2}
\begin{cases}
\DD  \frac{\partial h}{\partial s}\bigl(\lambda_n(d,\phi),\lambda_\mu(d,\phi)\bigr)= \frac{d}{1-d},\\[0.3cm]
\DD  \frac{\partial h}{\partial t}\bigl(\lambda_n(d,\phi),\lambda_\mu(d,\phi)\bigr) = \frac{\phi}{1-d}.
\end{cases}
\end{equation}
Inserting the variational solution (\ref{eq:varsol}) into $\cal F$ leads to the resulting expression 
of the rate function as the Legendre transform of $h$
\[
J(d,\phi) = \frac{d}{1-d}\lambda_n(d,\phi) + \frac{\phi}{1-d}\lambda_\mu(d,\phi)
-h\bigl(\lambda_n(d,\phi),\lambda_\mu(d,\phi)\bigr) 
\]
which could actually be obtained directly from  Cramer's theorem because $\{(n_i,\mu_i)\}$
are iid~\cite{Touch}. Concerning the other rate  $I(d)$ function, if $\lambda_n'(d)$ is associated to the first 
constraint, while $\lambda_\mu$ is set to zero we get 
\[
I(d) = \frac{d}{1-d}\lambda_n'(d) - h\bigl(\lambda_n'(d),0\bigr) 
\]
in term of the cumulant generating function $h$, which completes the determination of 
the large deviation FD (\ref{eq:ldfd}).

\subsection{Gaussian fluctuations around the fundamental diagram}\label{sec:gfluct}
The deterministic part of the fundamental diagram is then obtained for $\phi(d)$ such that, 
\begin{equation}\label{eq:djphi}
\frac{\partial J(d,\phi)}{\partial \phi}\Big|_{d,\phi(d)} = \frac{\lambda_\mu(d,\phi(d))}{1-d} = 0.
\end{equation}
When considering small fluctuations rather than large deviations, we recover a central 
limit theorem version of the fundamental diagram. In that case, replacing in (\ref{eq:ldzln})
$\cal F$ with its quadratic approximation around some stationary point value 
$y^*$ of $y$ yields the following expression for the partitions functions:
\begin{align}
Z_L(N) &\simeq 
\frac{\exp\bigl(-LI(d)\bigr)\bigr)}
{\sqrt{2\pi L h_{ss}({\lambda'}_n,0)}}\label{eq:ZLN}\\[0.2cm] 
Z_L(N,\Phi) &\simeq \frac{\exp\bigr(-LJ(d,\phi)\bigl)}
{\sqrt{(2\pi L)^2 \det\bigl(H^\star(\lambda_n,\lambda_\mu)\bigr)}}\label{eq:ZLNPhi} 
\end{align}
with $(\lambda_n,\lambda_\mu)$ and ${\lambda'}_n$ solving the constraints (\ref{eq:constraints2}) 
for a given value of $d$ and $\phi$ and denoting
\begin{equation}\label{def:Hdual}
H^\star(s,t) \egaldef 
\left[\begin{matrix}
h_{ss} & h_{st}\\
h_{ts} & h_{tt}
\end{matrix}\right],
\end{equation}
with the use of shorthand notations for the derivatives, 
to represent the covariant matrix between the charges of the queues and the flux.
This matrix is actually the Hessian in the dual representation, 
associated to the Legendre transform of the free energy (see Appendix~\ref{sec:appA} for details).
We obtain for the  large deviation FD the expression
\[
P(\phi|d) = \sqrt{\frac{L}{2\pi (1-d)^2} 
\frac{h_{ss}(\lambda_n',0)}{\det\bigl(H^\star(\lambda_n,\lambda_\mu)\bigr)}}
\exp\bigr(-LK(\phi\vert d)).
\]
The Taylor expansion at second order in $\phi-\phi(d)$ at fixed density is performed, 
using (\ref{eq:djphi}) so that  
\[
J\bigl(d,\phi(d)\bigr) = I(d).
\]
and 
\[
\frac{\partial^2 J(d,\phi)}{\partial \phi^2}\vert_{d,\phi(d)} =
\frac{1}{1-d}\frac{\partial \lambda_\mu(d,\phi)}{\partial \phi}\vert_{d,\phi(d)} =
\frac{1}{(1-d)^2}\frac{h_{ss}(\lambda_n,\lambda_\mu)}{\det(H^\star(\lambda_n,\lambda_\mu))}. 
\]
The last equality is obtained by remarking that from  the constraints (\ref{eq:constraints2})
we have 
\begin{align*}
\frac{\partial \lambda_n}{\partial\phi}h_{ss}+ 
\frac{\partial \lambda_\mu}{\partial\phi}h_{st} &= 0\\[0.2cm]
\frac{\partial \lambda_n}{\partial\phi}h_{st}+ \frac{\partial \lambda_\mu}{\partial\phi}h_{tt}
&= \frac{1}{1-d},
\end{align*}
from which $\partial \lambda_\mu/\partial\phi$ can be eliminated to yield the result.
Therefore, 
\[
K(\phi\vert d) = \frac{\bigl(\phi-\phi(d)\bigr)^2}{2(1-d)^2}\frac{h_{ss}}{\det(H^\star(\lambda_n,\lambda_\mu)}
+o\bigl(\phi-\phi(d)\bigr)^2.
\]
giving the Gaussian fluctuations of the FD.
The variance as a function of the density then reads
\begin{equation}\label{eq:FDvar}
\Var(\phi|d) = \frac{(1-d)^2}{L} \frac{\det(H^\star(\lambda_n(d),0)}{h_{ss}(\lambda_n(d),0)} = 
\frac{(1-d)^2}{L} 
\bigl({H^{\star-1}}_{tt}\bigr)^{-1}.
\end{equation}

\subsection{Special case: the TASEP}

Let us apply this formula first to the simple M/M/1 file corresponding
to one single speed level ($\mu_a=\mu_b=\mu$), i.e.\ 
a TASEP process. The rate of arrival $\lambda$ is set by convenience to $\lambda=d \mu$. 
The cumulant-generating function then reads
\[
h(s,t) = \log(1-d) + \log\frac{1+d e^{-s}(e^{-tv}-1)}{1-d e^{-s}}.
\]
Therefore we have 
\[
\det(H^\star(0,0)) = \left|
\begin{matrix}
\DD \frac{d}{(1-d)^2} & d \mu \\[0.2cm]
\DD d \mu & d(1-d)\mu^2
\end{matrix}\right|
=  \frac{d^3}{1-d}\mu^2,
\]
while $\E[\phi|d] = \mu d(1-d)$. Equation (\ref{eq:FDvar}) yields
\begin{equation}\label{eq:varvar}
\Var_{\text{var}}(\phi| d) = \frac{1-d}{L}d^2(1-d)^2 \mu^2.
\end{equation}
The direct result, in the grand canonical ensemble, with  expectation 
constraint on the density is given by
\[
\Var_{\text{GC}}(\phi|d) = \frac{1-d}{L}h_{tt} = \frac{1-d}{L}d(1-d). 
\]
Actually, the TASEP computation on the ring geometry can be done
directly as
\begin{align*}
\frac{1}{\mu^2}\text{Var}(\phi|d) &= \text{Var}\Bigl(\frac{1}{N+L} 
\sum_{i=1}^{N+L}\tau_i\bar\tau_{i+1}\Bigr)\\[0.2cm]
&=\frac{(L+N)(L+N-3)}{(L+N)^2}\E(\tau_i\bar\tau_{i+1}\tau_j\bar\tau_{j+1})\\[0.2cm]
&+\frac{1}{L+N}\E(\tau_i\bar\tau_{i+1})-\frac{1}{(L+N)^2}\Bigl(\E(\tau_i\bar\tau_{i+1})\Bigr)^2,
\end{align*}
with $\bar\tau_i =1-\tau_i \in\{0,1\}$, $\tau_{N+L+1} = \tau_1$ and
$\forall(i,j)$ s.t. $\vert i-j\vert\ge 2$. The expressions above are
computed as to give
\[
\E(\tau_i\bar\tau_{i+1}) = \frac{NL}{(L+N)(L+N-1)} = d(1-d)(1+\frac{1}{L+N}) 
+o\bigl(\frac{1}{L+N}\bigr)
\]
and 
\begin{align*}
\E(\tau_i\bar\tau_{i+1}\tau_j\bar\tau_{j+1}) &= 
\frac{N(N-1)L(L-1)}{(L+N)(L+N-1)(L+N-2)(L+N-3)}\\[0.2cm] 
&= d^2(1-d)^2(1+\frac{6}{L+N})-\frac{1}{L+N}d(1-d)+o\bigl(\frac{1}{L+N}\bigr),
\end{align*}
which finally yields (\ref{eq:varvar}).

\section{Queues with two-state service rates}\label{sec:single-jam}
Among the possible mappings discussed in Section~\ref{sec:qmap},
taking empty sites as queues should give more information on the jam structure, i.e.\ presumably
on the long range correlations of the model
associated to cluster formations, than the mapping of type (i). 
Unfortunately, the mapping of type (ii) is not exact in this case. Nevertheless, let us 
try to define a queueing process able to capture these correlations at least qualitatively.

The first step that we take in this direction is to understand 
the dynamics and steady state regime of a single cluster 
of vehicles of size $n_t$, assuming that vehicles of type $A$ [resp.\ $B$]  
join the queue with rate $\lambda_a$
[resp.\ $\lambda_b$], while they leaves the queue with rate $\mu_a$  [resp.\ $\mu_b$].
Let $\lambda\egaldef \lambda_a+\lambda_b$ represent the intensity of the incoming process.
In the bulk of the queue, accordingly to the transitions rules defined in 
Section~\ref{sec:model}, only deceleration is possible, with rate $\delta$.

\subsection{Speed profile inside a single cluster}
To study the speed profile, i.e. the speed labels of particle depending on their position inside
a cluster, let us encode a given cluster configuration as a binary sequence $\{(A_0,\ldots,A_{n-1}\}$, 
with $A_i = 1-B_i, \in\{0,1\}$, corresponding to either a fast or a slow vehicle. 
$n$ represents the size of the queue, the front vehicle has index $n-1$ 
while index $0$ corresponds to the last entered one. Let us remark first that 
the front end of the queue is a moving interface, which next $n$ moves are 
completely conditioned by the present  state of the sequence. 
Let us forget about this interface for the 
moment and consider the distribution of the infinite sequence  $\{(A_1,A_2,A_3,\ldots\}$. The joint 
law has not a product form, because the whole sequence is shifted 
by one unit ($A_i\lra A_{i+1}$)  at instants   
of arrival, which causes dependencies between neighbors. Nevertheless, the structure of dependency is 
quite simple, for example single site marginals are determined in close form. 
Let $p_i(t) = P_t(A_i=1)$, 
we have
\begin{align*}
\frac{d p_i}{dt} &= \lambda(p_{i-1}-p_i) -\delta p_i,\qquad\forall i> 1 \\[0.2cm]
\frac{d p_1}{dt} &= \lambda_a(1-p_1)-(\lambda_b+\delta)p_1,
\end{align*}
with $\lambda\egaldef \lambda_a+\lambda_b$.
The generating function
\[
\phi_t(z) \egaldef \sum_{i=0}^\infty p_{i+1}(t)z^i,
\]
satisfies
\[
\frac{d\phi_t(z)}{dt} = \bigl(\lambda z-(\lambda+\delta)\bigr)\phi_t(z)+\lambda_a,
\]
with solution 
\[
\phi_t(z) = \frac{\lambda_a+\phi_0(z)e^{-(\lambda+\delta-\lambda z)t}}{\lambda+\delta-\lambda z}.
\]
Correlations between nearest neighbors,  
$q_i(t)\egaldef P_t(A_i=1,A_{i+1}=1)$ satisfy  a similar closed form equation:
\begin{align*}
\frac{d q_i}{dt} &= \lambda(q_{i-1}-q_i) -2\delta q_i,\qquad\forall i\ge 1 \\[0.2cm]
\frac{d q_1}{dt} &= \lambda_ap_1(t)-(\lambda+2\delta)q_1,
\end{align*}
with 
\[
p_1(t) = \frac{\lambda_a+p_1(0)e^{-(\lambda+\delta)t}}{\lambda+\delta},
\]
which generating function 
\[
\psi_t(z) \egaldef \sum_{i=0}^\infty q_i(t)z^i,
\]
satisfies
\[
\frac{d\psi_t(z)}{dt} = \bigl(\lambda z-(\lambda+2\delta)\bigr)\psi_t(z)+\lambda_ap_1(t),
\]
having for solution:
\begin{align*}
\psi_t(z) &= \frac{\lambda_a^2}{(\lambda+\delta)(\lambda+2\delta-\lambda z)}\\[0.2cm]
&+ \bigl(\psi_0(z)-\frac{\lambda_a^2}{(\lambda+\delta)(\lambda+2\delta-\lambda z)}+
\frac{\lambda^2}{(\lambda+\delta)\lambda z}p_1(0)(1-e^{-\lambda zt})\bigr)e^{-(\lambda+2\delta-\lambda z)t}.
\end{align*}
At steady state we see that 
\[
p_i^\infty = \frac{\lambda_a}{\lambda}
\Bigl(\frac{\lambda}{\lambda+\delta}\Bigr)^i\qquad\text{and}\qquad  
q_i^\infty = \frac{\lambda_a^2}{\lambda(\lambda+\delta)}
\Bigl(\frac{\lambda}{\lambda+2\delta}\Bigr)^i,
\]
indicating a significant level of correlations when $\delta/\lambda$ is $O(1)$.
\begin{figure}
\begin{center}
\resizebox*{0.8\textwidth}{!}{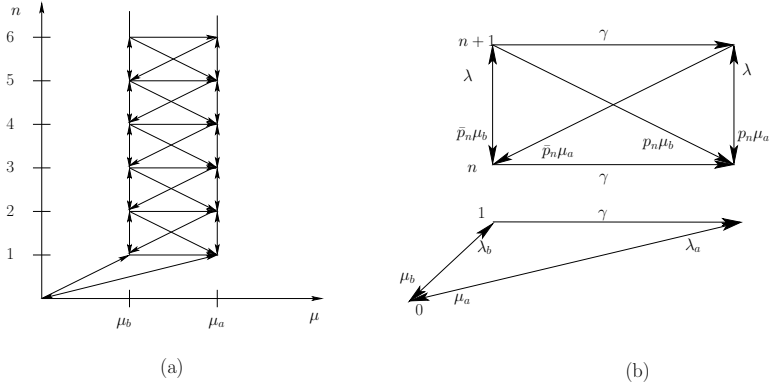}
\end{center}
\caption{(a) State flow diagram of a queue with stochastic service rate. 
(b) Details of the transition rates.}\label{fig:singlejam}
\end{figure}

\subsection{Joined effective process $\bigl(n(t),\mu(t)\bigr)$}
Since the front end interface of the cluster has no causal effect on the rest of the queue, 
except for the front vehicle which may accelerate with rate $\gamma$, we can 
consider the dynamics of the sequence independently of the motion of the front interface.
Making the additional assumption of independence of the local speed labels 
in the bulk is rather crude regarding the results of the previous section, but 
nevertheless we look for a qualitative comparison.
To this end let us  write a master equation of the joint process $(n(t),\mu(t))$, i.e.
the equation governing the  evolution of the joint 
probability $P_t(n,\tau) = P(n(t)=n,\mu=\mu_a\tau+\mu_b\bar\tau)$, 
the joint probability that the queue has
$n$ clients and its front car is of type $A$ ($\tau=1$) or $B$ ($\bar\tau\egaldef 1-\tau =1$).
Given
\[
p_n \egaldef \frac{\lambda_a}{\lambda}r^n\qquad\text{with}
\qquad r\egaldef \frac{\lambda}{\lambda+\delta}.
\]
and $p_n(\tau) \egaldef p_n\tau+\bar p_n\bar\tau$, using $\bar p_n$ to denote $1-p_n$,
the master equation then reads
\begin{align*}
\frac{dP_t(n,\tau)}{dt} &= \lambda \bigl(P_t(n-1,\tau)-P_t(n,\tau)\bigr)
+ \bigl(\mu_aP_t(n+1,1)+\mu_bP_t(n+1,0)\bigr)p_n(\tau)\\[0.2cm]
&-(\mu_a\tau +\mu_b\bar\tau) P_t(n,\tau)
+\gamma(\tau -\bar\tau) P_t(n,0),\qquad n\ge 2
\end{align*}
\begin{align*}
\frac{dP_t(1,\tau)}{dt} &= (\lambda_a\tau+\lambda_b\bar\tau) 
P_t(0)-\lambda P_t(1,\tau)
+ \bigl(\mu_aP_t(2,1)+\mu_bP_t(2,0)\bigr)p_1(\tau)\\[0.2cm]
&-(\mu_a\tau +\mu_b\bar\tau )P_t(1,\tau)
+\gamma(\tau -\bar\tau) P_t(1,0),
\end{align*}
and
\[
\frac{dP_t(0)}{dt} = -\lambda P_t(0)
+ \mu_aP_t(1,1)+\mu_bP_t(1,0).
\]
It is a special case of a queueing process  with a $2$-level dynamically 
coupled stochastic service rate,
as defined in Section~\ref{sec:stoch-queues},
which state-graph is sketch on Figure~\ref{fig:singlejam}.
In the stationary regime, we denote
\[
\pi_{n}^{a} \egaldef P(n(t)=n,\mu = \mu_a)\qquad\text{and}\qquad \pi_{n}^{b} 
\egaldef P(n(t)=n,\mu=\mu_b)
\]
and $\pi_0 = P(n(t)=0)$.

\begin{prop}\label{prop:recurrence}
For $n>1$ the following recursion holds
\begin{equation}\label{eq:recurrence}
\left[\begin{matrix}
\pi_{n+1}^a\\[0.2cm]
\pi_{n+1}^b
\end{matrix}\right] 
= \frac{\lambda}{\mu_a\mu_b +\mu_a(\gamma+\lambda p_{n+1})+\lambda \mu_b\bar p_{n+1}}
\left[\begin{matrix}
\gamma+\mu_b+\lambda p_{n+1} &  \gamma+\lambda p_{n+1} \\[0.2cm]
\lambda\bar p_{n+1} & \mu_a+\lambda\bar p_{n+1}
\end{matrix}\right]
\left[\begin{matrix}
\pi_n^a\\[0.2cm]
\pi_n^b
\end{matrix}\right] 
\end{equation}
while for $n=1$ we have 
\begin{align*}
\pi_1^a &= \frac{\lambda_a \mu_b+\lambda\gamma+\lambda^2 p_1}{\mu_a\mu_b+\lambda p_1\mu_a
+\lambda \bar p_1\mu_b+\gamma \mu_a}\pi_0,\\[0.2cm]
\pi_1^b &= \frac{\lambda_b \mu_a+\lambda^2\bar p_1}
{\mu_a\mu_b+\lambda p_1\mu_a+\lambda\bar p_1\mu_b+\gamma \mu_a}\pi_0.
\end{align*}
\end{prop}
\begin{proof} See Appendix~\ref{sec:appB}\end{proof}

Note that for large $n$, because of a finite acceleration rate $\gamma$, 
the effective service rate $\mu^{\infty}$ has a limit slightly above $\mu_b$. It is obtained as
\[
\mu^{\infty} = \mu_b + \frac{\eta}{1+\eta}(\mu_a- \mu_b),
\]
where $\eta$ is the limit ratio
\[
\eta = \lim_{n\to\infty}\frac{\pi_n^a}{\pi_n^b}
=\frac{1}{2\lambda}\bigl(\sqrt{\Delta}+\gamma-\lambda+\mu_b-\mu_a\bigr)
\qquad\text{with}\qquad \Delta \egaldef (\lambda-\gamma+\mu_a-\mu_b)^2+4\lambda\gamma
\]
obtained from (\ref{eq:recurrence}).
\[
\lambda = \sum_{n=1}^\infty\pi_n^a \mu_a +\pi_n^b \mu_b
\]
is automatically fulfilled by virtue of the partial balance equation 
(\ref{eq:pbalance}). 
Consider the generating functions 
\[
g_{a,b}(z) \egaldef \sum_{n=1}^\infty \pi_n^{a,b} z^n \qquad\text{and}\qquad g(z) 
\egaldef \pi_0+g_a(z)+g_b(z).
\]
\begin{prop}\label{theo:sol}
(i) $g(z)$ satisfies the functional equation of the type
\begin{equation}\label{eq:fxnal}
ug(rz) =\ -(z-z^+)(z-z^-)g(z)- vz+w 
\end{equation}
where $z^\pm$ are given by
\[
z^\pm = \frac{1}{2\lambda}\bigl(\mu_a+\mu_b+\lambda+\gamma \pm\sqrt{\Delta}\bigr),
\]
and
\begin{align*}
\lambda^2u = \lambda_a(\mu_a-\mu_b)\qquad
\lambda^2v = (\lambda_a\mu_a+\lambda_b \mu_b)\pi_0\qquad
\lambda^2w = (\mu_a\mu_b+\lambda_a \mu_a+\lambda_b \mu_b+\gamma \mu_a)\pi_0. 
\end{align*}
(ii) The solution reads: 
\begin{equation}\label{eq:sol}
g(z) = \sum_{n=0}^\infty (-u)^n\frac{w-vr^n z}{\prod_{k=0}^n(zr^k-z^+)(zr^k-z^-)}.
\end{equation}
\end{prop}
\begin{proof}See Appendix~\ref{sec:appB} for details.\end{proof}

\noindent The reduced generating functions $g_{a,b}$ are then obtained using 
(\ref{eq:pbalance}) to yield
\begin{equation}\label{eq:stoch-gen}
g_a(z) = \frac{\mu_b}{\mu_a-\mu_b}\pi_0 +\frac{\lambda z - \mu_b}{\mu_a-\mu_b}g(z),
\qquad\text{and}\qquad
g_b(z) = \frac{\mu_a}{\mu_b-\mu_a}\pi_0 +\frac{\lambda z - \mu_a}{\mu_b-\mu_a}g(z).
\end{equation}
From these it is then straightforward to obtain the $\pi_n^{a,b}$ upon using Cauchy integrals, as
sums of geometric laws.

\subsection{Mean field estimates}
If we consider then the closed tandem formed out of these effective queuing processes, 
the first observation is that the sufficient conditions given in section~\ref{sec:prodform}
may not be fulfilled. As such the large deviation estimation of the FD is going to 
be only approximate, a kind of mean field approximation, 
for a second reason, in addition to the simplifying assumption 
made on the internal cluster structure. 
\begin{figure}
\centering
\includegraphics[scale=0.3]{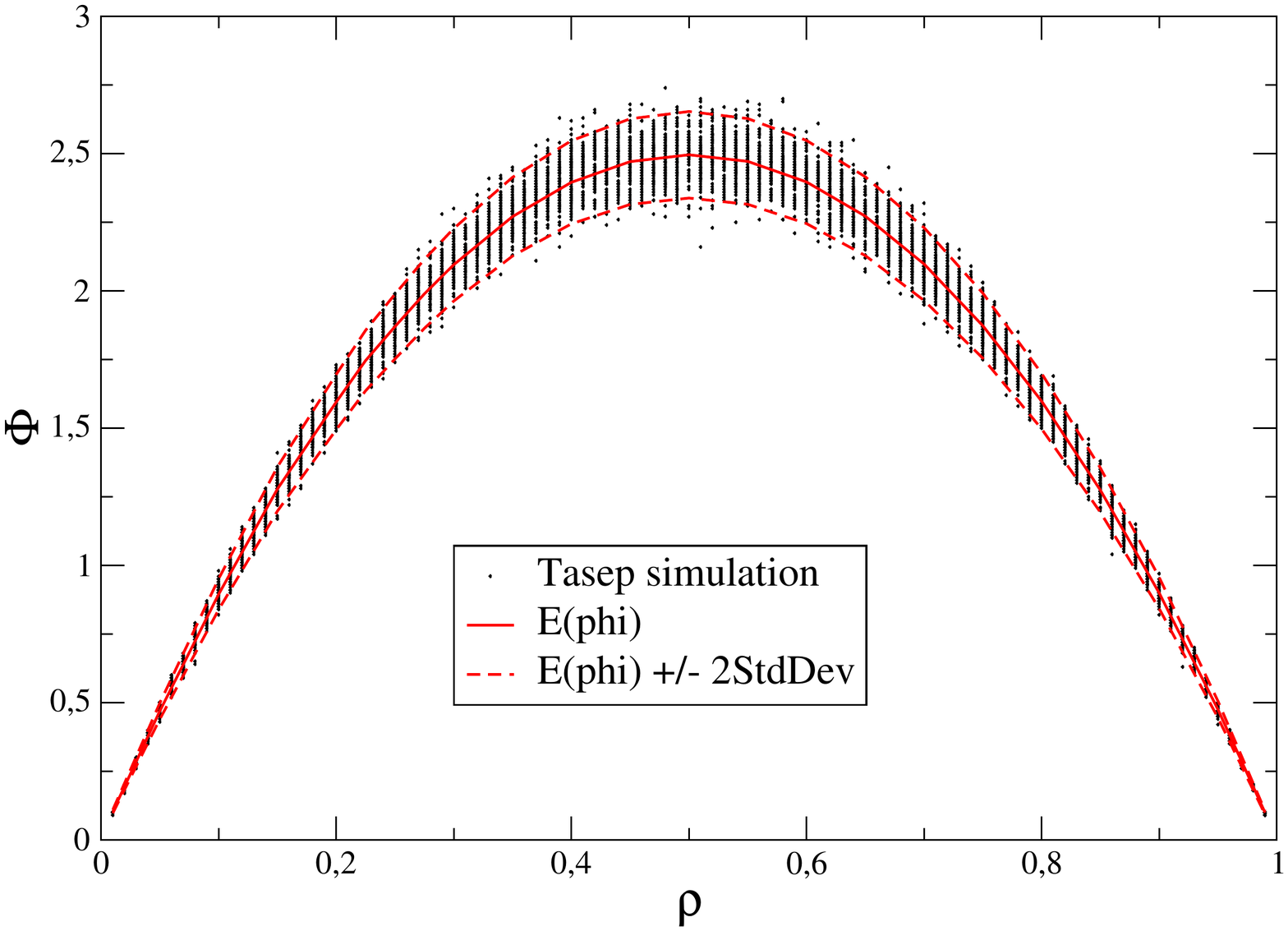}\\
(a)\\
\includegraphics[scale=0.3]{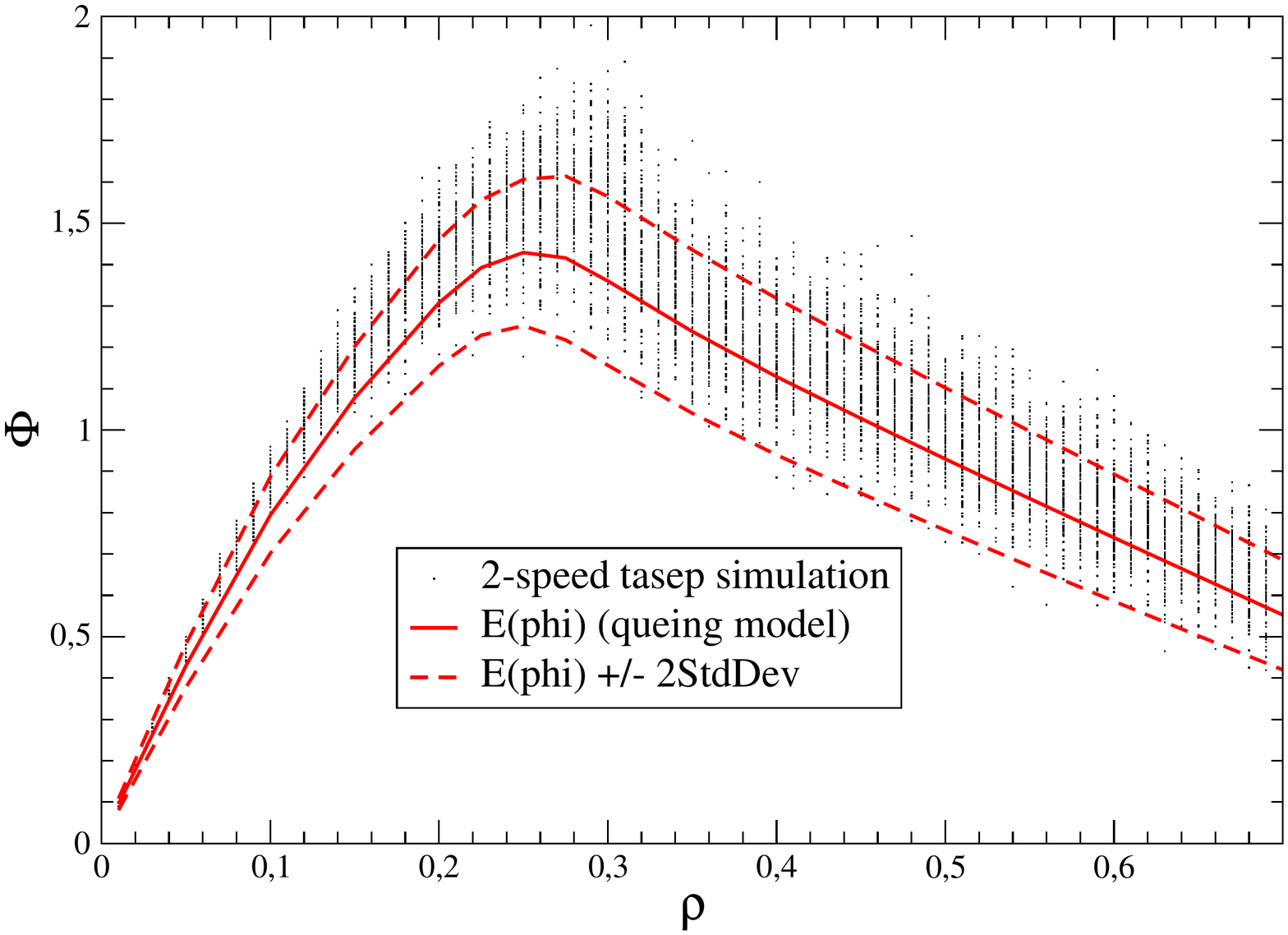}\\
(b)\\
\caption{\label{fig:fdiag} 
Comparison between the fundamental diagram  simulated with exclusion processes and the one 
obtained from the corresponding 
(approximate) queuing process. The size is $L=1000$, $\delta$ is set to zero to have a TASEP (a),
while the parameters are set to $\mu_a=10\times \mu_b=10\times \gamma=100\times \delta$ on panel (b).}
\end{figure}
\begin{figure}
\centering
\includegraphics[trim=0 0 100 0,width=0.7\linewidth]{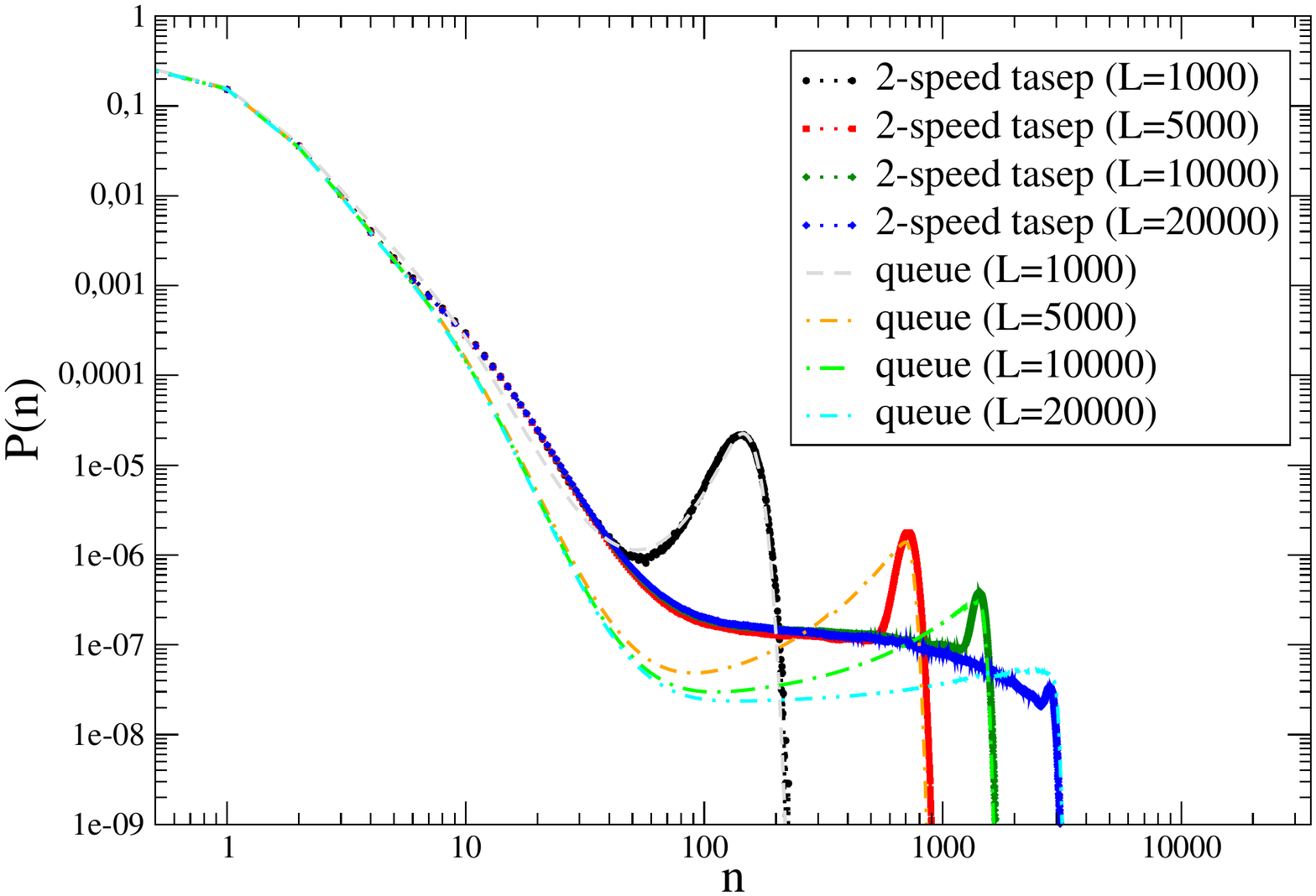}\\
(a)\\
\includegraphics[trim=0 0 100 0,width=0.7\linewidth]{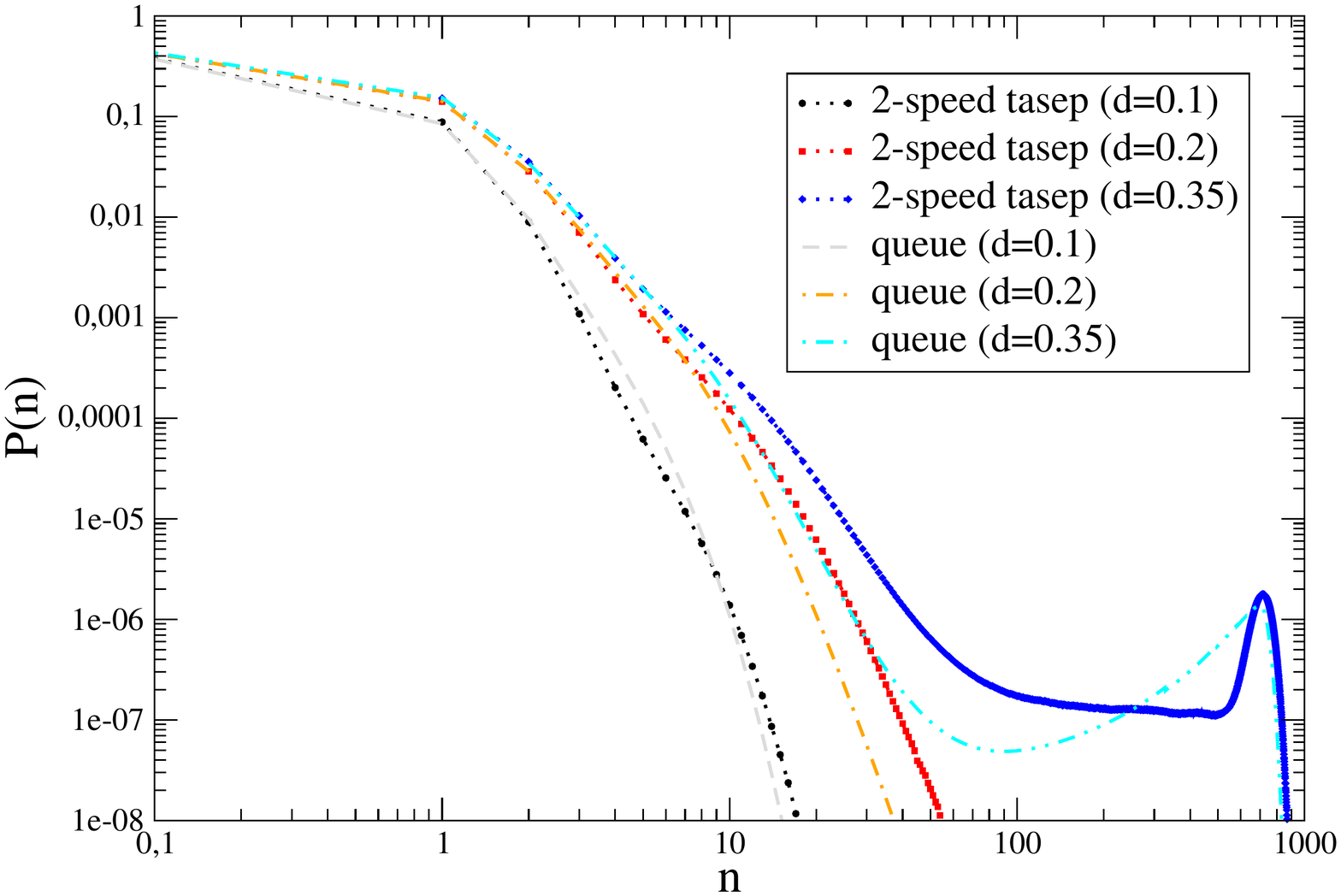}\\
(b)
\caption{\label{fig:mtasep-queue} Comparison of particle cluster vs queue's size  
distribution for set of 
parameters $\mu_a=10 \mu_b =100$, $\gamma = 10$ and $\delta=1$ for various sizes $L$ with fixed 
density $d=0.35$ (a) and various densities with fixed size $L=5000$ (b), the number of queues
being $(1-d)L$.}
\end{figure}

From (\ref{eq:stoch-gen}), we write down the cumulant generating function
\[
h(s,t) = \log\Bigl(\pi_0\bigl(1+\frac{\mu_be^{-\mu_at}-\mu_ae^{-\mu_bt}}{\mu_a-\mu_b}\bigr)+
\frac{(\mu_a-\lambda e^{-s})e^{-\mu_bt}+(\lambda e^{-s}-\mu_b)e^{-\mu_at}}{\mu_a-\mu_b}\ g(s)\Bigr).
\]
Therefore we get for the Hessian
\begin{align*}
&\det\bigl(H(s^\star,0)\bigr) = \\[0.2cm]
&\frac{1}{g^2}\left[
\begin{matrix}
\DD gg''-g'^2 & g\bigl(\lambda e^{-s^\star}g +(1-\lambda e^{-s^\star})g'\bigr) \\[0.2cm]
\DD g\bigl(\lambda e^{-s^\star}g +(1-\lambda e^{-s^\star})g'\bigr) & g\bigl( \mu_a\mu_b \pi_0 + 
(\lambda(\mu_a+\mu_b)e^{-s^\star}-\lambda^2-\mu_a\mu_b)g\bigr)
\end{matrix}\right],
\end{align*}
where $s^\star$ is the point which satisfies
\[
g'(s^\star) = -\rho.
\]
from this we can compute the small (Gaussian) fluctuations of the FD.
A comparison with  the corresponding ABTASEP is displayed on Figure~\ref{fig:fdiag}. 

In addition to the fundamental diagram, under the canonical ensemble constraint,
it is also interesting to determine  the single queue distribution. 
This can be obtained from the partition function~\cite{EvMaZi} in the large deviation framework as
\begin{align*}
p_{CE}(n,\mu) &= \pi^\lambda(n,\mu)\frac{Z_{L-1}(N-n)}{Z_L(N)}\\
\simeq  \pi^\lambda(n&,\mu)\exp\Bigl[L\bigl(h(s(d-x),0)-h(s(d),0)-
\frac{d-x}{1-d-x}s(d-x)+\frac{d}{1-d}s(d)\bigr)\Bigr],
\end{align*}
with $x\egaldef n/(N+L)$ and the density constraint
\[
\frac{\partial h}{\partial s}(s(d),0) = \frac{d}{1-d},
\]
satisfied by $s(d)$. A comparison of this queueing formulation with
the ABTASEP is given on Figure~\ref{fig:mtasep-queue}. In this
case we impose an additional self-consistency condition on the
parameter $\lambda_a$,
\[
\lambda_a  = \sum_{n=1}^\infty\pi_n^a \mu_a,
\]
which otherwise would be free. The correspondence between the cluster size
distribution observed on the ABTASEP with the single queue distribution
obtained from the generalized queueing process is rather accurate. In particular in both 
cases, a bump is observed in the distributions at the same location, for small
size systems. It indicates that condensation is observed as a finite size phenomena. 
In the thermodynamic limit macroscopic jams are absent. In this respect 
it is different from the type of condensation analyzed in~\cite{EvMaZi}, which are 
obtained under some conditions on the service rate, as a large deviation principle 
but with different scaling (speed in the large deviation terminology) then $L$.
Concerning the FD a larger discrepancy is observed between the ABTASEP process and 
its corresponding effective queueing process. The reason for this, which is not visible 
on the cluster distributions, is that the  ratio of fast over slow particles for
given size of cluster size  does not coincide. We don't know however which
approximation between 
\begin{itemize}
\item the simplified assumption on internal cluster structure
\item the neglect of correlations between queues
\end{itemize}
in the effective model is responsible for that. These results indicate anyway 
that the steady state of the ABTASEP might be well accounted for 
by a little bit more refined joint cluster measure.

\section{Conclusion and Perspectives}
In this paper, motivated by questions raised in the context of traffic modelling, 
we have proposed a simple extension in the definition of the TASEP model to take into 
account acceleration and  braking thereby offering the possibility to study the 
effect of asymmetry between the two mechanisms within a simple model. 
With two different speed levels a rich dynamics is already present. 
An effective mapping on a generalized zero-range process allows
to interpret at least qualitatively the spontaneous jamming phenomena occurring 
in this model for some choice of parameters. 
In addition we develop a large deviation formalism to study the fundamental
diagram associated to this family of queuing processes. Still the discrepancy between 
numerical simulation and estimation seen on Figures~\ref{fig:mtasep-queue} and \ref{fig:fdiag}, 
may originate from the various approximations we make, concerning the detail structure
of the jams and the correlations between queues which are neglected when using our large 
deviation estimation of the FD. This formalism could be possibly adapted to 
the situation where the joint measure of queues has not a product form. 

In the family of models that we considered  we found that condensation phenomena are associated to
finite size effect on the ring geometry, and at least numerically the presence of a large 
macroscopic jam at large scale seems doomed to decay exponentially with no bump in the 
probability distribution, and this is consistent with the approximate mapping on 
zero-range process that we propose. Nevertheless the situations could change when the 
number of speed levels is increased, as shown in our previous work~\cite{tgf07} on the 
subject, and we suspect that synchronized flow can take place in this limit of 
large number of speed level. 

The situation where vehicles may enter or leave the system
at very low rates so that the density may change adiabatically and allow  
to observed hysteresis effect deserves also more studies. In particular it would be interesting
to find a relevant measure associated to trajectories in the FD, based for example 
on Brownian  windings models proposed in~\cite{FaFu}, to quantify the hysteresis level 
in various regions of the FD.
 
Finally, the models considered here are limited to single lane traffic
and could be easily generalized to multi-lane, using coupled exclusion processes
like in~\cite{EvKaSuTa}.

\appendix
\section{Complement to Section~\ref{sec:gfluct}: constrained partition function
and the dual Hessian}\label{sec:appA}
In this appendix we explain the role played by the dual Hessian in the 
Large deviation expression of the constrained partition functions.
Assume that we have $n$ linear constraints, written as 
\[
Cy = V,
\]
where $C$ is $n\times d$ matrix, each line $C_i$ corresponds
to constraints $i$ and $V$ is a $n$-dimensional  vector  and that we want to estimate 
the constrained partition function $Z_L[V]$. 
To obtain the small fluctuations we approximate 
first $\cal F$ in (\ref{eq:ldzln}) at second order with respect to some reference
point $y^*$
\[
{\cal F}(y) = {\cal F}(y^*) + (y-y^*)^T\cdot\bigtriangledown{\cal F}(y^*) 
+(y-y^*)^T\bigtriangledown^2{\cal F}(y^*)(y-y^*)+o\bigl(\|y-y^*\|^2\bigr),
\]
so that $y$ considered as a $d$-dimensional vector is integrated over ${\mathbb R}^d$.
$y^*$ is chosen such that 
\[
y^*(\Lambda) = \argmin_y\ \Bigl[{\cal F}[y]- \Lambda^T Cy\Bigr],
\] 
matches the constraints with the proper Lagrange multipliers $\Lambda$.
Since in absence of constraints $Z_L[V]$, as it stands in (\ref{eq:znl}) is  
normalized to $1$, so at this order of approximation  the partition function 
reads 
\begin{align*}
Z_L[V] = L^{d/2}\sqrt{\frac{\det H}{(2\pi)^d}}\int d^dy  
\ &e^{- L\bigr({\cal F}(y^*) + (y-y^*)^T\cdot\bigtriangledown{\cal F}(y^*)+\frac{1}{2}(y-y^*)^T H(y-y^*)\bigr)}\\[0.2cm]
&\times\prod_{i=1}^n\delta(V_i-\sum_{j=1}^dC_{ij}y_j),
\end{align*}
with 
\[
H = \bigtriangledown^2{\cal F}(y^*),
\]
the Hessian taken at $y=y^*$. The dual Hessian is associated to the dual free energy ${\cal F}^\star$,
\[
{\cal F}^\star[\Lambda] = {\cal F}[y^*(\Lambda)]-\Lambda^T Cy.
\]
Using relations associated to the stationarity of $y^*$, it reads
\[
H^\star = C^TH^{-1}C.
\]
Due to the specific form (\ref{eq:freen}) of ${\cal F}$, the 
Hessian turns out to be the covariance matrix between 
the various quantities $V_i$, as given in (\ref{def:Hdual}). 
A convenient way to express the constraints is to write
\[
\prod_{i=1}^n\delta(V_i-\sum_{j=1}^dC_{ij}y_j) = \lim_{\alpha\to 0}
\frac{1}{(2\pi\alpha)^{n/2}} e^{-\frac{1}{2\alpha} \|V-Cy\|^2},
\]
so that a Gaussian integration over $y$ can be performed and since
\[
\|V-Cy\|^2 = \|C(y-y^*\|^2,
\]
we simply get 
\[
Z_L[V] = e^{- L{\cal F}(y^*)}
\lim_{\alpha\to 0} \sqrt{\frac{\det H}{(2\pi\alpha)^n\det H_\alpha}}
e^{\frac{L}{2}V_\alpha^T H_\alpha^{-1} V_\alpha}
\]
with 
\begin{align*}
H_\alpha &\egaldef  H+\frac{C^TC}{L\alpha}\\[0.2cm]
V_\alpha &\egaldef  \bigtriangledown{\cal F}(y^*).
\end{align*}
Let $P$ be the orthonormal projection on the subspace spanned by the set of 
constraint vectors $C_k$ an $\bar P\egaldef 1-P$, such that 
\[
C = CP\qquad\text{and}\qquad P\ C^T = C^T. 
\]
We have
\begin{align*}
\det(H_\alpha) &= \det(H)\det\bigl(1+\frac{C^TC}{L\alpha}H^{-1}\bigr)\\[0.2cm]
&=\det(H)\det\bigl(P(1+\frac{C^TC}{L\alpha}H^{-1})P+ P\frac{C^TC}{L\alpha}H^{-1}\bar P+\bar P\bar P\bigr)\\[0.2cm]
&=\det(H)\det_P\bigl(1+\frac{C^TC}{L\alpha}H^{-1}\bigr)\\[0.2cm]
&=\det(H)\Bigl(\frac{1}{(L\alpha)^n}\det(H^\star)+o\bigl(\frac{1}{\alpha^n}\bigr)\Bigr)
\end{align*}
where $\det_P$ corresponds to the block determinant associated to the subspace of constraints, 
and 
\[
H^\star = C^T H^{-1} C
\]
is the dual Hessian defined in (\ref{def:Hdual}).
Concerning $V_\alpha$, since $y^*$ is a stationary point at least w.r.t. fluctuations
orthogonal to the constraints, so we have
\[
\bar P V_\alpha = 0.
\]
Since
\begin{align*}
H_\alpha^{-1} &= H^{-1}\bigl(1+\frac{C^TC}{L\alpha}H^{-1}\bigr)^{-1} \\[0.2cm]
& = H^{-1}\bigl(PL\alpha H(C^TC)^{-1}P-P\bar P+\bar P\bar P\bigr)+o(\alpha),
\end{align*}
as a result 
\[
\lim_{\alpha\to 0}\ V_\alpha^T H_\alpha^{-1} = 0.
\]
So finally the Gaussian estimate of the partition function reads
\[
Z_L[V] = \frac{L^{n/2}}{\sqrt{(2\pi)^n\det(H^\star)}}e^{-L{\cal F}(y^*)}
\]
which yield in particular the expressions (\ref{eq:ZLN},\ref{eq:ZLNPhi}) up to a 
factor $L^n$ caused by a different convention in the constraint definition.

\section{Complement to Section~\ref{sec:single-jam}: steady state of the queueing process}\label{sec:appB}
From the master equations, taken at steady-state, we get the recurrence for $n\ge 2$
\begin{align*}
\lambda\bigl(\pi_{n-1}^a-\pi_n^a\bigr)+\bigl(\mu_a\pi_{n+1}^a+\mu_b\pi_{n+1}^b\bigr)p_n
-\mu_a\pi_n^a+\gamma\pi_n^b = 0,\\[0.2cm]
\lambda\bigl(\pi_{n-1}^b-\pi_n^b\bigr)+\bigl(\mu_a\pi_{n+1}^a+\mu_b\pi_{n+1}^b\bigr)\bar p_n
-\mu_b\pi_n^b-\gamma\pi_n^b = 0.
\end{align*}
For $n=1$, it reads
\begin{align*}
\lambda_a\pi_0-\lambda\pi_1^a+\bigl(\mu_a\pi_2^a+\mu_b\pi_2^b\bigr)p_1
-\mu_a\pi_1^a+\gamma\pi_1^b = 0,\\[0.2cm]
\lambda_b\pi_0-\lambda\pi_1^b+\bigl(\mu_a\pi_2^a+\mu_b\pi_2^b\bigr)\bar p_1
-\mu_b\pi_1^b-\gamma\pi_1^b = 0.
\end{align*}
The sum of the two equation gives that the quantity 
$
\lambda \pi_n - \bigl(\mu_a\pi_{n+1}^a+\mu_b\pi_{n+1}^b\bigr)
$
is a constant independent of $n$ which has to vanish, 
leading to the partial balance relation
\begin{equation}\label{eq:pbalance}
\lambda\pi_n = \mu_a\pi_{n+1}^a+\mu_b\pi_{n+1}^b,\qquad\forall n\ge 0.
\end{equation}
Using this gives the  relations, 
\begin{align}
\bigl(\mu_a+\lambda\bar p_{n+1}\bigr)\pi_{n+1}^a- (\gamma+\lambda p_{n+1})\pi_{n+1}^b &= 
\lambda\pi_n^a \label{eq:invreca}\\[0.2cm] 
-\lambda\bar p_{n+1}\pi_{n+1}^a + (\gamma+\mu_b+\lambda p_{n+1})\pi_{n+1}^b &= \lambda\pi_n^b.\label{eq:invrecb}
\end{align}
The recurrence is then inverted to yield~(\ref{eq:recurrence}).

Recurrence (\ref{eq:invreca},\ref{eq:invrecb}) are then conveniently used to obtain
the following relations among the generating functions, 
\begin{align*}
\lambda_a g_a(rz) + \lambda_a g_b(rz) &= 
(\mu_a+\lambda-\lambda z)g_a(z)- \gamma g_b(z) - \lambda_a \pi_0 z\\[0.2cm]
\lambda_a g_a(rz) + \lambda_a g_b(rz) &= \lambda g_a(z)- (\gamma+\mu_b-\lambda z) g_b(z) + \lambda_b\pi_0 z.
\end{align*}
It follows that
\[
g_b(z) = \frac{\lambda z}{\mu_b - \lambda z}\pi_0 + \frac{\lambda z-\mu_a}{\mu_b - \lambda z} g_a(z), 
\]
also related to (\ref{eq:pbalance}), which rewrites
\[
g_a(z) = \frac{\mu_b}{\mu_a-\mu_b}\pi_0 +\frac{\lambda z - \mu_b}{\mu_a-\mu_b}g(z),\qquad\text{and}\qquad
g_b(z) = \frac{\mu_a}{\mu_b-\mu_a}\pi_0 +\frac{\lambda z - \mu_a}{\mu_b-\mu_a}g(z).
\]
So finally $g(z)$ satisfies the functional equation given in proposition~\ref{theo:sol}.
The solution is constructed as follows. Consider the infinite product
\[
G(z) \egaldef \prod_{n=0}^\infty \bigl(1-r^n\frac{z}{z^+}\bigr)\bigl(1-r^n\frac{z}{z^-}\bigr).
\]
We have
\begin{equation}\label{eq:grz}
G(rz) = \frac{z^-z^+}{(z-z^-)(z-z^+)}G(z), 
\end{equation}
so a solution of the form
\[
g(z) = \frac{C(z)}{G(z)}
\]
has to verify
\[
C(z)+ \frac{u}{z^+z^-}C(rz) = \frac{w-vz}{(z-z^-)(z-z^+)}G(z).
\]
Equivalently, for any $n\ge 0$ this rewrites
\begin{align*}
\bigl(-\frac{u}{z^+z^-}\bigr)^n C(r^nz)- \bigl(-\frac{u}{z^+z^-}\bigr)^{n+1} C(r^{n+1}z) 
&= \bigl(-\frac{u}{z^+z^-}\bigr)^n \frac{w-vzr^n}{(zr^n-z^-)(zr^n-z^+)}G(r^nz)\\[0.2cm]
&= (-u)^n\frac{w-vzr^n}{\prod_{k=0}^n(zr^k-z^-)(zr^k-z^+)}G(z),
\end{align*}
after using relation (\ref{eq:grz}). 
Taking the sum leads to the solution~(\ref{eq:sol}), owing to
\[
\lim_{n\to\infty}\bigr(-\frac{u}{z^+z^-}\bigl)^nC(r^n z) = 0,
\]
which can be checked afterwards.

\paragraph{Acknowledgments} 
We thank Guy Fayolle for useful discussions. 
This work was supported by the French National Research Agency (ANR) grant No ANR-08-SYSC-017.

{\small \bibliography{fdiag}}
\bibliographystyle{unsrt}

\end{document}